\documentclass[11pt]{article}
\usepackage{fullpage}

\usepackage{latexsym}
\usepackage{amsmath}
\usepackage{amssymb}
\usepackage{amsthm}
\usepackage{hyperref}
\usepackage{cite}
\usepackage{graphicx}
\usepackage{color}
\usepackage{framed}
\usepackage{xspace}
\usepackage{subfigure}

\newtheorem{theorem}{Theorem}
\newtheorem*{metaPostulate}{Meta-postulate}
\newtheorem{postulate}{Postulate}
\newtheorem{model}{Model}
\newtheorem{theorem*}{Theorem}
\newtheorem{corollary}[theorem]{Corollary}
\newtheorem{lemma}[theorem]{Lemma}
\newtheorem{observation}[theorem]{Observation}
\newtheorem{proposition}[theorem]{Proposition}
\newtheorem{definition}[theorem]{Definition}
\newtheorem{claim}[theorem]{Claim}
\newtheorem{fact}[theorem]{Fact}

\theoremstyle{definition}
\newtheorem{problem}{Problem}

\newtheorem{remark}[theorem]{Remark}

\renewcommand{\vec}[1]{\mathbf{#1}}

\newcommand{\GL}{\mathrm{GL}}
\newcommand{\Iso}{\mathrm{Iso}}
\newcommand{\RIso}{\mathrm{RIso}}

\newcommand{\IndIso}{\mathrm{IndIso}}
\newcommand{\Aut}{\mathrm{Aut}}
\newcommand{\RAut}{\mathrm{RAut}}
\newcommand{\F}{\mathbb{F}}
\newcommand{\Z}{\mathbb{Z}}

\newcommand{\C}{\mathbb{C}}
\newcommand{\N}{\mathbb{N}}

\newcommand{\trans}[1]{{#1}^{\mathrm{t}}}

\newcommand{\rk}{\mathrm{rk}}

\newcommand{\im}{\mathrm{im}}

\newcommand{\poly}{\mathrm{poly}}

\newcommand{\cB}{\mathcal{B}}
\newcommand{\cC}{\mathcal{C}}
\newcommand{\cG}{\mathcal{G}}
\newcommand{\cH}{\mathcal{H}}

\newcommand{\bB}{\mathbf{B}}
\newcommand{\bC}{\mathbf{C}}
\newcommand{\bG}{\mathbf{G}}
\newcommand{\bH}{\mathbf{H}}

\newcommand{\gbinom}[3]{{\genfrac{[}{]}{0pt}{}{#1}{#2}}_{#3}}
\newcommand{\Adj}{\mathrm{Adj}}
\newcommand{\LinER}{\algprobm{LinER}}
\newcommand{\BipLinER}{\algprobm{BipLinER}}
\newcommand{\ER}{\algprobm{ER}}
\newcommand{\NaiT}{\algprobm{NaiT}}
\newcommand{\NaiS}{\algprobm{NaiS}}
\newcommand{\BipNaiT}{\algprobm{BipNaiT}}
\newcommand{\BipNaiS}{\algprobm{BipNaiS}}

\newcommand{\algprobm}[1]{\mbox{\sc #1}\xspace}
\newcommand{\AltMatSpIso}{\algprobm{AltMatSpIso}}
\newcommand{\GrI}{\algprobm{GraphIso}}
\newcommand{\GpI}{\algprobm{GroupIso}}
\renewcommand{\P}{{\rm P}\xspace}
\newcommand{\NP}{{\rm NP}\xspace}
\newcommand{\coAM}{{\rm coAM}\xspace}
\newcommand{\zerovec}{\mathbf{0}}

\title{
Linear algebraic analogues of \\
the graph isomorphism problem
and the Erd\H{o}s-R\'enyi model
}

\author{
Yinan Li \thanks{Centre for Quantum Software and Information, 
 University of Technology Sydney, Australia  ({\tt liyinan9252@gmail.com}).}
\and 
Youming Qiao \thanks{Centre for Quantum Software and Information, 
 University of Technology Sydney, Australia  ({\tt jimmyqiao86@gmail.com}).}
}

\date{\today}

\begin{document}

\pagenumbering{gobble}  

\maketitle

\begin{abstract}
A classical difficult isomorphism testing problem is to test isomorphism of 
$p$-groups of class $2$ and exponent $p$ in time polynomial in the group order. It 
is known that this problem can be reduced to solving the alternating matrix 
space isometry problem over a finite field 
in time polynomial in the underlying 
vector space size. We propose a venue of attack for the 
latter problem by viewing 
it as a linear algebraic analogue of the graph isomorphism problem. This viewpoint 
leads us to explore the possibility of transferring techniques for graph 
isomorphism to this long-believed bottleneck case of group isomorphism. 

In 1970's, Babai, Erd\H{o}s, and Selkow presented the first average-case 
efficient graph isomorphism testing algorithm (SIAM J Computing, 1980). Inspired 
by that algorithm, we devise an average-case efficient algorithm for the 
alternating matrix space isometry problem over a key range of parameters, in a 
random model of alternating matrix spaces in vein of the Erd\H{o}s-R\'enyi model 
of random graphs. For this, we develop a linear algebraic analogue of the 
classical individualisation 
technique, a technique belonging to a set of combinatorial techniques that has 
been critical for the progress on the worst-case time complexity for graph 
isomorphism, but 
was
missing in the group isomorphism context. As a consequence of the main algorithm, 
we establish a weaker linear algebraic analogue 
of Erd\H{o}s and R\'enyi's classical result that most graphs have the trivial 
automorphism group.
We also show that 
Luks' dynamic programming technique for graph isomorphism (STOC 1999) can be 
adapted to slightly improve the worst-case time complexity of the alternating 
matrix space isometry problem in a certain range of parameters.

Most notable progress on the worst-case time complexity of graph isomorphism, 
including 
Babai's recent breakthrough (STOC 2016) and Babai and Luks' previous record (STOC 
1983), has 
relied on both group theoretic and 
combinatorial techniques. By developing a linear algebraic analogue of the 
individualisation technique and demonstrating its usefulness in the average-case 
setting, the main result opens up the possibility of adapting that strategy for 
graph 
isomorphism to this hard instance of group isomorphism. 
The linear algebraic Erd\H{o}s-R\'enyi 
model is of independent interest and may deserve further study. In particular, we 
indicate a connection with enumerating $p$-groups of class $2$ and exponent $p$.
\end{abstract}

\newpage
\pagenumbering{arabic}     

\section{Introduction}\label{sec:intro}


\subsection{Problems, postulates, and models}

Let $\F_q$ be the finite field with $q$ elements. An $n\times n$ matrix $A$ over 
$\F_q$ is \emph{alternating}, if for every $u\in \F_q^n$, 
$\trans{u}Au=0$.
$\Lambda(n, q)$ denotes the 
linear 
space of $n\times n$ alternating matrices over $\F_q$, and a dimension-$m$ 
subspace of $\Lambda(n, q)$ is called an $m$-alternating (matrix) space.
$\GL(n, q)$ denotes the 
general linear group of degree $n$ over $\F_q$. 
We study the following problem.

\begin{problem}[Alternating matrix space isometry problem, 
$\AltMatSpIso$]\label{prob:main}
Given the 
linear bases of two $m$-alternating spaces $\cG, \cH$ in $\Lambda(n, q)$, decide 
whether there exists $A\in \GL(n, q)$, such that $\trans{A}\cG A:=\{ \trans{A}BA : 
B\in 
\cG\}$ is equal to $\cH$ as subspaces. 
\end{problem}

If such an $A$ exists, we say that $\cG$ and $\cH$ are \emph{isometric}. As will 
be 
explained in Section~\ref{subsec:bg}, $\AltMatSpIso$ has been studied, mostly 
under other names, for decades. It lies at the heart 
of the group isomorphism problem (\GpI), and has an intimate relationship with the 
celebrated 
graph isomorphism problem (\GrI). As a problem in $\NP\cap \coAM$, its 
worst-case time complexity has barely been 
improved over the brute-force algorithm. In fact, a $q^{O(n+m)}$-time algorithm 
is already regarded as very difficult. 

Let us recall one formulation of \GrI. For $n\in 
\N$, let $[n]=\{1, 2, \dots, n\}$, and $S_n$ 
denotes the symmetric group on $[n]$. A simple undirected graph is just a 
subset of $\Lambda_n:=\{\{i, j\} : i, j\in [n], i\neq j\}$. A permutation 
$\sigma\in S_n$ induces a natural action on $\Lambda_n$. The following formulation 
of \GrI as an instance of the setwise transporter problem is well-known 
\cite{Luk82}. 
\begin{problem}[Graph isomorphism problem, \GrI]\label{prob:gi}
Given two subsets $G, H$ of $\Lambda_n$, decide whether there 
exists $\sigma\in S_n$, such that $G^\sigma:=\{\{i^\sigma, j^\sigma\} : \{i, j\} 
\in G\}$ is equal to $H$ as sets.
\end{problem}
The formulations of \AltMatSpIso and \GrI as in Problem~\ref{prob:main} and 
Problem~\ref{prob:gi} lead us to the following postulate.
\begin{postulate}\label{post:lin_gi}
\AltMatSpIso can be viewed and studied as a linear algebraic analogue of \GrI.
\end{postulate}

Postulate~\ref{post:lin_gi} originates from the following meta-postulate. 
\begin{metaPostulate}
Alternating matrix spaces can be viewed and studied as a 
linear algebraic analogue of graphs. 
\end{metaPostulate}
This meta-postulate will be studied further in \cite{Qia17}. 
As a related note, recent progress on the non-commutative rank problem 
suggests the usefulness of viewing linear spaces of matrices as a linear algebraic 
analogue of bipartite graphs \cite{GGOW16,IQS16,IQS17}. 

From the meta-postulate, we formulate a model of random 
alternating matrix spaces over $\F_q$. Let $\gbinom{\ }{\ }{q}$ be the Gaussian 
binomial 
coefficient with base $q$. 
\begin{model}[The linear algebraic Erd\H{o}s-R\'enyi model]\label{model:linER}
The linear algebraic Erd\H{o}s-R\'enyi model, $\LinER(n, m, q)$, is the uniform
probability distribution 
over the set of dimension-$m$ subspaces of $\Lambda(n, 
q)$, that is, each subspace is endowed with probability 
$1/\gbinom{\binom{n}{2}}{m}{q}$.
\end{model}
Model~\ref{model:linER} clearly mimics the usual Erd\H{o}s-R\'enyi model 
\cite{ER59,ER63,Bol01}.
\begin{model}[Erd\H{o}s-R\'enyi model]\label{model:er}
The Erd\H{o}s-R\'enyi model $\ER(n, m)$ is the uniform probability distribution 
over the set of size-$m$ subsets of $\Lambda_n$, that is, each subset is 
endowed with probability $1/\binom{\binom{n}{2}}{m}$.
\end{model}
We then pose the following postulate.
\begin{postulate}\label{post:linER}
$\LinER(n, m, q)$ can be viewed and studied as a linear algebraic analogue of 
$\ER(n, m)$. 
\end{postulate}

\subsection{Background of the alternating matrix space isometry 
problem}\label{subsec:bg}




While the name \AltMatSpIso may be unfamiliar to some readers, this problem has 
been studied for decades as an instance -- in fact, the long-believed 
bottleneck case -- of the group isomorphism problem. This 
problem also has an intricate relationship with the graph isomorphism problem. We 
first 
review these connections below, and then examine the current status of this 
problem.

\paragraph{Relation with the group isomorphism problem.} We first 
introduce the 
group isomorphism problem (\GpI) and mention a long-believed bottleneck instance 
of this 
problem. It turns out that \AltMatSpIso is almost equivalent to this instance. 

\GpI  asks to decide whether two finite groups of order $n$ are isomorphic or not. 
The difficulty of this problem depends crucially on how we represent the groups in 
the algorithms. If our goal is to obtain an algorithm running in time $\poly(n)$, 
then we may assume that we have at our disposal the Cayley 
(multiplication) table of the group, as we can recover the Cayley table from most 
reasonable 
models for computing with finite 
groups. 
Therefore, in the main text we restrict our discussion to this very redundant 
model, 
which is 
meaningful mainly because we do not know a $\poly(n)$-time or even an $n^{o(\log 
n)}$-time algorithm \cite{Wil14} ($\log$ to the base $2$), despite that a simple 
$n^{\log n+O(1)}$-time 
algorithm has been known for decades \cite{FN70,Mil78}.
The past few years have witnessed a resurgence of activity on algorithms for this 
problem with worst-case analyses in terms of the group order; we refer the reader 
to \cite{GQ14} which 
contains a survey of these algorithms.

It is long believed that $p$-groups form 
the bottleneck case for \GpI. 
In fact, the decades-old quest for a polynomial-time algorithm has focused on 
class-$2$ $p$-groups, with little success. Even if we restrict further 
to consider $p$-groups of class $2$ and exponent $p$, the problem is still 
difficult. Recent works 
\cite{LW12,BW12,BMW15,IQ17} solve some nontrivial subclasses of this group class, 
and have lead to substantial improvement in practical algorithms. But the methods 
in those works seem not helpful enough to lead to any improvement for the 
worst-case time complexity of the general class. 

By a classical result of Baer \cite{Bae38}, testing isomorphism of $p$-groups of 
class $2$ and exponent 
$p$ in time polynomial in the group order reduces to solving \AltMatSpIso over 
$\F_p$ in time 
$p^{O(m+n)}$. On the other hand, there also is an inverse reduction for $p>2$. In 
fact, when such $p$-groups are given by generators in the 
permutation group quotient model \cite{KL90}, isomorphism testing reduces to 
solving 
\AltMatSpIso 
in time $\poly(n, m, \log p)$ \cite{BMW15}.  We will recall the reductions 
in Appendix~\ref{app:eq}. Because of these reductions and the current status of 
\GpI, we see that \AltMatSpIso lies at the heart of \GpI, and solving \AltMatSpIso 
in $q^{O(m+n)}$ is already very difficult. 

\paragraph{Relation with the graph isomorphism problem.} The 
celebrated graph 
isomorphism 
problem (\GrI) asks to decide whether two 
undirected simple graphs are isomorphic. The relation between \AltMatSpIso and 
\GrI is very delicate. Roughly speaking, the two time-complexity measures of 
\AltMatSpIso, $q^{O(n+m)}$ 
and $\poly(n, m, q)$, sandwiches 
\GrI in an interesting way. For one direction, solving \AltMatSpIso in time 
$q^{O(n+m)}$ can be reduced to solving 
\GrI for graphs of size $q^{O(n+m)}$, by first reducing to 
solving \GpI for groups of order $q^{O(n+m)}$ as above, and then to solving 
\GrI for graphs of size $q^{O(n+m)}$ by the reduction from \GpI to \GrI 
\cite{KST93}. 
Therefore, a polynomial-time algorithm for \GrI 
implies an algorithm for \AltMatSpIso in time $q^{O(n+m)}$. It is then reasonable
to examine whether the recent breakthrough of Babai \cite{Bab16,Bab17}, a 
quasipolynomial-time algorithm for \GrI, helps with reducing the time complexity 
of \AltMatSpIso. This seems unlikely. One indication is that the 
brute-force 
algorithm 
for \AltMatSpIso is already quasipolynomial with respect to $q^{O(n+m)}$. Another 
evidence 
is that Babai in \cite[arXiv version 2, Section 13.2]{Bab16} noted that his 
algorithm seemed not helpful to improve \GpI, and posed \GpI as one roadblock for 
putting \GrI in \P. Since \AltMatSpIso captures the long-believed bottleneck case 
for \GpI, the current results for \GrI are unlikely to improve the time 
complexity to $q^{O(n+m)}$. There is also an explanation from the technical 
viewpoint \cite{GR16}. Roughly speaking, the barrier in the group theoretic 
framework for \GrI is to deal with large alternating groups, as other composition 
factors like projective special linear groups can be handled by brute-force in 
quasipolynomial time, so for the purpose of a quasipolynomial-time algorithm these 
group are not a concern. On the other hand for \AltMatSpIso it is exactly the 
projective special linear groups that form a bottleneck. 
For the other 
direction, in a 
forthcoming work 
\cite{GQ17}, it is shown that solving \GrI in polynomial 
time reduces to solving \AltMatSpIso over $\F_q$ with $q=\poly(n)$ in time 
$\poly(n, m, 
q)$.  

\paragraph{Current status of \AltMatSpIso.}
It is not hard to show that solving \AltMatSpIso in time $\poly(n, m, \log q)$ is 
in 
$\NP\cap \coAM$, so it is unlikely to be $\NP$-complete. 
As to the worst-case time complexity, the brute-force 
algorithm for \AltMatSpIso runs in time 
$q^{n^2}\cdot \poly(m, n, \log q)$. 
Another analysed algorithm for \AltMatSpIso offers a running time of 
$q^{\frac{1}{4}(n+m)^2+O(n+m)}$ when $q=p$ is a prime, by first reducing to 
testing 
isomorphism of class-$2$ and 
exponent-$p$ $p$-groups of order $p^{n+m}$, and then applying Rosenbaum's 
$N^{\frac{1}{4}\log_p N+O(1)}$-time algorithm for $p$-groups of order 
$N$ \cite{Ros13a}. This is only better than the brute-force one when 
$m<n$.\footnote{As pointed out in \cite{BMW15}, there are numerous unanalysed 
algorithms \cite{OBr93,ELO02} which 
may lead to some improvement, but $q^{c n^2}\cdot \poly(n, m, \log q)$ for some 
constant $0<c<1$ is a reasonable over estimate of the best bound by today's 
method.}
It is somewhat embarrassing that for a problem in $\NP\cap \coAM$, 
we are only able to barely improve over the brute-force algorithm in a limited 
range of parameters.
In a very true sense, our current understanding of 
the worst-case time complexity of \AltMatSpIso is like the situation for \GrI in 
the 1970's. 

On the other hand practical algorithms for \AltMatSpIso have been implemented. As 
far as we know, current implemented algorithms for \AltMatSpIso can handle the 
case 
when $m+n\approx 20$ and $p\approx 13$, but absolutely not the case if $m+n\approx 
200$, though for $m+n\approx 200$ and say $p\approx 13$ the input can be stored 
in a 
few megabytes.\footnote{We thank James B. Wilson, who maintains a suite of 
algorithms for $p$-group isomorphism testing, for communicating his hands-on 
experience to us. We take the responsibility for any possible misunderstanding or 
not knowing of the performance of other implemented algorithms.}
For \GrI, the programs {\sc Nauty} and {\sc Traces} \cite{MP14} can test 
isomorphism of 
graphs stored in gigabytes in a reasonable amount of time. Therefore, 
unlike \GrI, \AltMatSpIso seems hard even in the practical sense. 

\paragraph{On the parameters.}
From the discussion above, we see that solving \AltMatSpIso with a worst-case time 
complexity $q^{O(n+m)}$ seems already a difficult target. From the meta-postulate, 
it is helpful to think 
of vectors in $\F_q^n$ as vertices, and matrices in an $m$-alternating space
as edges, so the $q^{O(n+m)}$ measure can be thought of as polynomial in the 
number of ``vertices'' 
and the number of ``edges.'' Here the parameter $m$ comes into the theme, because 
$q^m$, while no more than 
$q^{\binom{n}{2}}$, is not necessarily bounded by a polynomial in $q^n$.  
This is in contrast to \GrI, where the edge number is at most quadratic in the 
vertex number. 
In particular, when $m=\Omega(n^2)$, the brute-force algorithm which runs in 
$q^{n^2}\cdot \poly(m, n, \log q)$ is already in time $q^{O(n+m)}$. Furthermore, 
if we consider all $n\times n$ alternating matrix spaces (regardless of the 
dimension),  
most of them are of dimension $\Omega(n^2)$, so the 
brute-force algorithm already works in time $q^{O(n+m)}$ for most alternating 
matrix spaces. On the other hand, when $m$ is very small compared to $n$, say 
$m=O(1)$, we can enumerate all elements in $\GL(m, q)$ in time $q^{O(1)}$, 
and apply the isometry testing for alternating matrix \emph{tuples} from 
\cite{IQ17} which runs in randomized time $\poly(n, m, \log q)$. 
Therefore, the 
$q^{O(n+m)}$-time measure makes most sense when $m$ is comparable with $n$, in 
particular when $m=\Theta(n)$. This is why we study average-case 
algorithms in this regime of parameters (e.g. $\LinER(n, m, q)$ with 
$m=\Theta(n)$), 
while the average-case algorithm for \GrI in \cite{BES80} considers all graphs 
(e.g. each labelled graph is taken with probability $1/2^{\binom{n}{2}}$).

\subsection{Algorithmic results}

Postulates~\ref{post:lin_gi} and~\ref{post:linER} seem
hopeful at 
first sight by the formulations of \AltMatSpIso and \LinER. But realities in the 
combinatorial world and the linear algebraic world can be quite different, as just 
discussed in the last paragraph.  
So meaningful results cannot be obtained by adapting the results for 
graphs to alternating matrix spaces in a straightforward fashion. 
One purpose of this article is to provide evidence that, despite 
potential 
technical difficulties, certain ideas that have been developed for \GrI and \ER 
can be adapted to work with \AltMatSpIso and \LinER.

We will take a shortcut, by 
presenting one result that supports both postulates. In the graph setting, such a 
result is naturally an average-case efficient graph isomorphism testing algorithm 
with the average-case
 analysis done in the Erd\H{o}s-R\'enyi model. The first such algorithm was 
proposed by Babai, Erd\H{o}s and Selkow in 1970's \cite{BES80}, with follow-up 
improvements by Lipton \cite{Lip78}, Karp \cite{Kar79}, and Babai and Ku\v{c}era 
\cite{BK79}. 
Therefore we set to study average-case algorithms for \AltMatSpIso in the \LinER 
model. Inspired by the algorithm in \cite{BES80}, we show the following.
\begin{theorem}[Main result]\label{thm:main}
Suppose $m=cn$ for some constant $c$. There is an algorithm which, for almost but 
at most $1/q^{\Omega(n)}$ fraction of 
alternating matrix spaces $\cG$ in 
$\LinER(n, m, q)$, tests 
any alternating matrix space $\cH$ for isometry to $\cG$ in time $q^{O(n)}$. 
\end{theorem}

An important ingredient in Theorem~\ref{thm:main}, the utility of which should go 
beyond 
the average-case setting, is an adaptation of
the \emph{individualisation} technique for \GrI to \AltMatSpIso. 
We also realise a reformulation of the \emph{refinement} technique for \GrI as 
used in 
\cite{BES80} in the \AltMatSpIso setting. 
Individualisation and refinement are very influential 
combinatorial ideas for \GrI, have been crucial in the progress of 
the worst-case time complexity of \GrI, including Babai's 
recent breakthrough \cite{Bab16,Bab17}, but were missing in the \GpI context.
\begin{quote}{\it
The main contribution of this article to \AltMatSpIso is to initiate the use of 
the individualisation and refinement ideas for
\GrI in this problem.}
\end{quote}
Here, we note an interesting historical coincidence. Babai was the first to 
import the group theoretic idea to \GrI in 1979 \cite{Bab79}, by when the 
combinatorial techniques had been around for quite some time. On the 
other hand, we have an opposite situation for \AltMatSpIso: the relevant group 
theoretic tools have been the subject of intensive study for decades, 
while it is the combinatorial 
individualisation and refinement 
ideas that need to be imported. We do understand though, that there are valid 
reasons for people not having come to this before. For example, we would not have 
come to such ideas, if we restrict ourselves to solving \AltMatSpIso 
in time $\poly(n, m, \log q)$. In Section~\ref{subsec:discussion_algo}, we will 
reflect on the historical development on the worst-case complexity of \GrI, and 
discuss the prospect of getting a $q^{O(n^{2-\epsilon})}$-time algorithm for 
\AltMatSpIso. 



For an $m$-alternating 
space $\cG$ in $\Lambda(n, q)$, define the autometry
group of $\cG$, $\Aut(\cG)$ 
as $\{A\in \GL(n, q) : A^t\cG A=\cG\}$.
The proof of Theorem~\ref{thm:main} implies the following, which can be viewed as 
a weaker correspondence of the classical result that most graphs 
have trivial automorphism groups \cite{ER63}. 
\begin{corollary}\label{cor:main}
Suppose $m=cn$ for some constant $c$. All but $1/q^{\Omega(n)}$ fraction of 
alternating matrix spaces in $\LinER(n, m, 
q)$ have autometry groups of size $q^{O(n)}$. 
\end{corollary}
We observe that Corollary~\ref{cor:main} has certain consequences to the 
enumeration of finite $p$-groups of class $2$. For details see 
Section~\ref{subsec:lb}.

Finally, we provide another piece of evidence to support 
Postulate~\ref{post:lin_gi}, by adapting Luks' dynamic programming technique for 
\GrI \cite{Luks99} to \AltMatSpIso. In the \GrI setting, this technique improves 
the naive 
$n!\cdot \poly(n)$ 
time bound to the $2^{O(n)}$ time bound, which can be understood  as 
replacing the 
number of permutations $n!$ with the number of subsets $2^n$. In the linear 
algebraic setting the analogue would be to replace $\Theta(q^{n^2})$, the number 
of 
invertible matrices 
over $\F_q$, with the number of subspaces in $\F_q^n$ which is 
$q^{\frac{1}{4}n^2+O(n)}$. We show that this is indeed possible.
\begin{theorem}\label{thm:minor}
There exists a deterministic algorithm for \AltMatSpIso in time 
$q^{\frac{1}{4}(m^2+n^2)+O(m+n)}$.
\end{theorem}
Note that the quadratic term on the exponent of the algorithm in 
Theorem~\ref{thm:minor} is $\frac{1}{4}(m^2+n^2)$, slightly better than the one 
based on Rosenbaum's result \cite{Ros13a} which is $\frac{1}{4}(m+n)^2$. We 
stress though that our 
intention to present this result is to support 
Postulate~\ref{post:lin_gi}.

\paragraph{Organisation of this paper.} In 
Section~\ref{sec:outline}, we explain the basic idea the algorithm for 
Theorem~\ref{thm:main}, by drawing analogues with the algorithm in \cite{BES80}. 
Then, after presenting 
preliminaries and preparation material, we present detailed proofs for 
Theorem~\ref{thm:main} (Section~\ref{sec:main_algo}) and Theorem~\ref{thm:minor} 
(Section~\ref{sec:dp}).
Section~\ref{sec:discussion} includes discussions, 
future directions, and connections to group enumeration.

\section{Outline of the main algorithm}\label{sec:outline}

We now describe the outline of  the algorithm for Theorem~\ref{thm:main}, which  
is inspired by the first average-case efficient algorithm for \GrI by Babai, 
Erd\H{o}s, and Selkow \cite{BES80}. We will recall the idea in \cite{BES80} that 
is relevant to us, define a linear algebraic individualisation, and 
propose a reformulation of the refinement step in \cite{BES80}. Then we 
present an outline of the main algorithm. 
%
During the procedure we will also see how the meta-postulate guides the 
generalisation here.

\subsection{A variant of the naive refinement algorithm as used in 
\cite{BES80}}\label{subsec:ir}

Two properties of random graphs are used in the average-case analysis of the 
algorithm in \cite{BES80}. The 
first property is that most graphs have the first $\lceil3\log n\rceil$
largest degrees distinct. The second property, which is relevant to us, is the 
following.

Let $G=([n], E)$ be a simple and undirected graph. Let $r=\lceil3\log n\rceil$, 
and $S=[r]$, 
$T=[n]\setminus [r]$. Let $B$ be the bipartite graph induced by the cut $[r]\cup 
\{r+1, \dots, n\}$, that is, $B=(S\cup T, F)$ where $F=\{(i, j) :  
i\in S, j\in T, \{i, j\} \in E\}$. For each $j\in T$, assign a length-$r$ bit 
string $f_j$ as follows: 
$f_j\in \{0, 1\}^r$ such that $f_j(i)=1$ if and only if $(i, j)\in F$. It is 
easy to verify that, all but at most $O(1/n)$ fraction of graphs 
satisfy that $f_j$'s are distinct over $j\in T$. 

Let us see how the second property alone, together with the
individualisation and refinement heuristic, 
give an average-case algorithm 
in $n^{O(\log n)}$. Suppose $G$ satisfies the second property, and we would like 
to 
test isomorphism between $G=([n], E)$ and an arbitrary graph $H=([n], E')$. Let 
$St_G\subseteq \{0, 
1\}^r$ be the set of bit strings obtained in the procedure above. Note that 
$|St_G|=n-r$.
In the \emph{individualising} step, we enumerate all $r$-tuple of vertices in $H$. 
For a fixed $r$-tuple $(i_1, 
\dots, i_r)\in [n]^r$, we perform the \emph{refinement} step, that is, label the 
remaining vertices in $H$ according to their 
adjacency relations with the $r$-tuple $(i_1, \dots, i_r)$ as before, to obtain 
another set of bit-strings $St_H$. If $St_G\neq St_H$ we neglect this $r$-tuple. 
If 
$St_G=St_H$, then form a bijective map between $[n]$ and $[n]$, by mapping $j$ to 
$i_j$ for $j\in[r]$, 
and the rest according to their labels. Finally check whether 
this bijective map induces an isomorphism. 

It can be verified easily that the above algorithm is an $n^{O(\log n)}$-time 
algorithm that tests isomorphism between $G$ and $H$ given that $G$ satisfies the 
required property. In particular, this 
implies that for such $G$, $|\Aut(G)|\leq n^{O(\log n)}$. To recover the algorithm 
in \cite{BES80}, assuming 
that 
the largest $r$ degrees are distinct, one can canonicalise the choice of the 
$r$-tuples by choosing the one with largest $r$ degrees for both $G$ and $H$.


\subsection{Individualisation and refinement in the \AltMatSpIso 
setting}\label{subsec:ind}

We will generalise the above idea to the setting of 
\AltMatSpIso. To do this, we first make sense of what individualisation means in 
the alternating space setting. We discuss how the refinement step may be 
generalised, and indicate how we follow an alternative formulation of it.

Let $G=([n], E)$ and $H=([n], E')$ be two graphs 
for which we want to test isomorphism. Let $\cG, \cH\leq \Lambda(n, q)$ be two 
$m$-alternating spaces for which we want to test isometry. As the case in 
Section~\ref{subsec:ir}, we will look for 
properties of $G$ or $\cG$ which enable the average-case analysis, and perform 
individualisation on $H$ or $\cH$ side.

For $i\in[n]$, $e_i$ denotes the $i$th standard basis vector of $\F_q^n$. For a 
vector space $V$ and $S\subseteq V$, we use $\langle S\rangle$ to denote the 
linear span of $S$ in $V$.

\medskip \noindent{\bf Individualisation.} In the graph setting, individualising 
$r$ vertices 
in $H$ can be 
understood as follows. First we fix a size-$r$ subset $L$ of $[n]$. Then put an 
order on the elements in $L$. The result is a tuple of distinct vertices $(i_1, 
\dots, i_r)\in [n]^r$. Enumerating such tuples incurs a multiplicative cost of at 
most $n^r$.

In the alternating matrix space setting, it is helpful to think of vectors in 
$\F_q^n$ as 
vertices, and matrices in $\cH$ as edges. Consider the following procedure.
First fix a dimension-$r$ subspace $L$ of $\F_q^n$. Then 
choose an ordered basis of $L$. The result is a tuple of linearly independent 
vectors 
$(v_1, \dots, v_r)$, $v_i\in \F_q^n$, such that $L=\langle v_1, \dots, 
v_r\rangle$. This incurs a 
multiplicative cost of at most $q^{rn}$. Up to this point, this is in complete 
analogy 
with the graph setting. We may stop here and say that an $r$-individualisation 
amounts to fix an $r$-tuple of linearly independent vectors.

We can go a bit further though. As will be clear in the following, it is 
beneficial if we also fix a complement 
subspace $R$ of $L$, e.g. $R\leq \F_q^n$ such that $L\cap R=\zerovec$ and $\langle 
L\cup R\rangle =\F_q^n$. This 
adds another multiplicative cost of $q^{r(n-r)}$, which is the number of 
complement subspaces of a fixed dimension-$r$ subspace in $\F_q^n$. In the graph 
setting, this step is not necessary, because for any $L\subseteq [n]$ there 
exists a unique complement subset $R=[n]\setminus L$. 

To summarise, by an 
\emph{$r$-individualisation}, we mean choosing a direct sum 
decomposition $\F_q^n=L\oplus R$ where $\dim(L)=r$ and $\dim(R)=n-r$, together 
with an ordered basis $(v_1, \dots, v_r)$ of $L$.
Enumerating all 
$r$-individualisations incurs a total multiplicative 
cost of at most $q^{2rn-r^2}$.

\paragraph{Towards a refinement step as in \cite{BES80}.} In the \GrI 
setting, 
individualising $r$ vertices gives $(i_1, \dots, i_r)\in[n]^r$, and allows us to 
focus on
isomorphisms that respect this individualisation, namely those $\phi\in\Iso(G, H)$ 
such 
that $\phi(j)=i_j$ for $j\in[r]$. There are at most $(n-r)!$ such isomorphisms. 
Since $r$ is 
usually set as a polylog, just naively trying all such permutations does not help. 
Therefore the 
individualisation is usually accompanied with a refinement type technique. 
Specifically, setting $L=\{i_1, \dots, i_r\}$ and $R=[n]\setminus L$, the 
refinement step as in \cite{BES80} assigns 
every $v\in R$ a label according to its adjacency relation w.r.t. $(i_1, \dots, 
i_r)$. 
This 
label in fact represents a \emph{subset} of $L$, and an 
individualisation-respecting isomorphism has to preserve this adjacency relation 
for every $v\in R$. 
This restriction turns out to be 
quite 
severe for most graphs: as observed in Section~\ref{subsec:ir}, for most graphs 
$G$, the adjacency relations between $(1, 2, \dots, r)$ and $j\in [n]\setminus 
[r]$ are completely different over $j$. For such $G$ and any individualisation of 
$H$, this means that there is at most one way to extend $\phi(j)=i_j$ for 
$j\in[r]$ to an isomorphism between $G$ and $H$. 


In the \AltMatSpIso setting, an $r$-individualisation also allows us to focus on 
isometries that respect the decomposition $L\oplus R$ and the 
ordered basis $(v_1, \dots, v_r)$ of $L$, namely those $\phi\in\Iso(\cG,\cH)$ such 
that 
$\phi(e_i)=v_i$ 
for $i\in[r]$, 
and $\phi(\langle e_{r+1}, \dots, e_n\rangle)=R$. There are at most 
$q^{(n-r)^2}$ such isometries. Since 
$r$ will be also set to be very small --  in fact a constant here -- we also need 
some refinement type argument. 
For $u\in R$, we can record 
its ``adjacency
relation'' w.r.t. $\vec{v}=(v_1, \dots, v_r)$ as a \emph{subspace} of 
$L\cong 
\F_q^r$ as follows. For $Q\in \cH\leq \Lambda(n, q)$, define $Q(\vec{v}, 
u):=(\trans{v}_1Qu, \dots, 
\trans{v}_rQu)^{\mathrm{t}}\in \F_q^r$, and $\cH(\vec{v}, u):=\{Q(\vec{v}, u) : 
Q\in \cH\}$ 
which is a 
subspace in $\F_q^r$. $\cH(\vec{v}, u)$ records the adjacency
relation between $(v_1, \dots, v_r)$ and $u$ under $\cH$. It can be verified that 
an individualisation-respecting isometry has to preserve this adjacency 
relation. 
It is tempting to check then on the $\cG$ side, where we have the standard 
individualisation $(e_1, \dots, 
e_r)$ and $\langle 
e_{r+1}, \dots, e_n\rangle$, whether for most 
$\cG$'s it is the case that every $v\in \langle 
e_{r+1}, \dots, e_n\rangle$ gets a unique label. If this is so, then the number 
of individualisation-respecting 
isomorphisms can also be significantly reduced. However, this cannot be the case 
when $r$ is small, as there are 
$q^{(n-r)^2}$ vectors in $R$ but there at at most $q^{r^2}$ subspaces in $\F_q^r$. 

The alert reader will note that, since we are looking for linear maps from 
$\langle 
e_{r+1}, \dots, e_n\rangle$ to $R$, the 
above counting argument does not make much sense, as it mostly concerns setwise 
maps from 
$\langle e_{r+1}, \dots, e_n\rangle$ to $R$. It is indeed the case, and we 
further note that the map from $u\in R$ to $\cH(\vec{v}, u)\leq \F_q^r$ defines a 
sheaf over the projective space $\mathbb{P}(R)$, 
so such labels have some 
nontrivial relation to glue 
together to form a sheaf. (See the related concept of kernel sheaves as in 
\cite{KV12}.) It may be possible to use these observations to define a 
reasonable refinement step in the alternating matrix space setting. In this paper 
we shall follow the following reformulation. 

\paragraph{A reformulation of the refinement step.} To resolve the above 
problem, we 
reformulate the idea 
in the graph setting as follows. Recall that on the $G$ side we start with the 
standard
individualisation $[r]\cup \{r+1, \dots, n\}$ with an order on $[r]$ as 
$(1, \dots, r)$, and let $S=[r]$, $T=\{r+1, \dots, n\}$. This defines the 
bipartite 
graph $B=(S\cup T, F)$ 
where the edge set $F$ is induced from $G$. For a fixed individualisation on the 
$H$ side, which produces $L\cup R$, $L=\{i_1, \dots, i_r\}\subseteq [n]$ with an 
order on $L$, this also defines a bipartite graph $C=(L\cup R, F')$ where $F'$ is 
induced from $H$. A bijective $\psi: T\to R$ is a right-side isomorphism between 
$B$ and $C$ if it induces an isomorphism between $B$ and $C$ as bipartite graphs. 
Let $\RIso(B, C)$ be 
the set of right-side isomorphisms, and let $\IndIso(G, H)$ be the set of 
individualisation-respecting isomorphisms from $G$ to $H$ w.r.t the above 
individualisations. Note that both $\RIso(B, C)$ and $\IndIso(G, H)$ can be 
embedded to the set of 
bijective maps between $T$ and $R$. The key observation is that an 
individualisation-respecting 
isomorphism has to be a right-side isomorphism between $B$ and $C$, e.g.  
$\IndIso(G, H)\subseteq \RIso(B, C)$. Also note that either $|\RIso(B, C)|=0$ 
(e.g. when $B$ and $C$ are not right-isomorphic), or 
$|\RIso(B, C)|=|\RAut(B)|$ 
where $\RAut(B):=\RIso(B, B)$. The refinement step as in 
Section~\ref{subsec:ir} achieves two goals. Firstly on the $G$ side, most $G$'s 
have the corresponding $B$ with $|\RAut(B)|=1$. This means that $|\RIso(B, C)|\leq 
1$. Secondly, given $H$ with a fixed individualisation inducing the 
corresponding bipartite graph $C$, there is an 
efficient procedure to decide whether $B$ and $C$ are right-isomorphic (by 
comparing the labels), and if they do, enumerate all right-isomorphisms (actually 
unique).

In the \AltMatSpIso 
setting, on the $\cG$ side we start with the standard 
individualisation $S=\langle e_1, \dots, e_r\rangle$, $T=\langle e_{r+1}, \dots, 
e_n\rangle$ with the ordered basis $(e_1, \dots, e_r)$ of $S$. We can also define 
a correspondence of the bipartite graph $B$ in this setting, which is the matrix 
space 
$\cB'=\{[e_1, \dots, e_r]^tP[e_{r+1}, \dots, e_n] : P\in \cG\}\leq M(r\times 
(n-r), q)$, where $[e_1, \dots, e_r]$ denotes the $n\times r$ matrix listing the 
column vectors $\{e_i:i=1,\dots,r\}$. Note that $[e_1, \dots, e_r]^tP[e_{r+1}, \dots, e_n]$ is 
just the 
upper-right $r\times (n-r)$ submatrix of $P$. Similarly, the individualisation on 
the $H$ side yields $L\oplus R$ with an ordered basis of $L$, $(v_1, \dots, v_r)$, 
$v_i\in \F_q^n$. Take any basis of $R=\langle v_{r+1}, \dots, v_n\rangle$. 
Similarly construct $\cC'=\{[v_1, \dots, v_r]^t Q [v_{r+1}, \dots, v_n] : Q\in 
\cH\}\leq M(r\times (n-r), q)$. $A\in\GL(n-r, q)$ is a right-side equivalence 
between $\cB'$ and $\cC'$ if $\cB' A:=\{DA : D\in \cB'\}=\cC'$. Let $\RIso(\cB', 
\cC')$ be the set of right-side equivalences between $\cB'$ and $\cC'$, and 
$\IndIso(\cG, \cH)$ the set of individualisation-respecting isometries between 
$\cG$ and $\cH$. Similarly, both $\RIso(\cB', \cC')$ and $\IndIso(\cG, \cH)$ can 
be 
embedded in the set of invertible linear maps from $T$ to $R$ (isomorphic to 
$\GL(n-r, q)$), and we have $\IndIso(\cG, \cH)\subseteq \RIso(\cB', \cC')$. 
Furthermore $\RIso(\cB', \cC')$ is either empty (e.g. $\cB'$ and $\cC'$ are not 
right-side equivalent), or a coset of 
$\RAut(\cB'):=\RIso(\cB', \cB')$. So in analogy with the graph 
setting, for our 
purpose the goals become: (1) 
for 
most $m$-alternating space $\cG\leq \Lambda(n, q)$ with $m=cn$ for some constant 
$c$, setting $r$ to be some constant, we 
have $|\RAut(\cB')|\leq q^{O(n)}$, and (2) for $\cG$'s satisfying (1), 
$\RIso(\cB', 
\cC')$ can be enumerated 
efficiently. 

We are almost ready for the algorithm outline. Alas, there is one important 
ingredient 
missing. It 
turns out for the purpose of (2), we will need to``linearise'' $\RAut(\cB')$ to 
allow for the use of efficient linear algebra procedures. This linearisation 
is captured by the adjoint algebra concept, defined
below in the algorithm outline. Correspondingly, in the goals above we will 
replace 
$\RAut(\cB')$ and $\RIso(\cB', \cC')$ with $\Adj(\bB)$ and $\Adj(\bB, \bC)$ where 
$\bB$ and $\bC$ will be defined below as well.

\subsection{Algorithm outline}\label{subsec:outline}

Suppose we want to test isometry 
between two $m$-alternating spaces $\cG=\langle G_1, \dots, G_m\rangle$ and 
$\cH=\langle H_1, \dots, H_m\rangle$ in $\Lambda(n, q)$. To ease the 
presentation in this subsection we assume $r=4$ and $m=n-4$. 

We first define the property on $\cG$ for the sake of average-case analysis. 
Given those $G_k\in \Lambda(n, q)$ linearly spanning $\cG$, form a $3$-tensor 
$\bG\in \F_q^{n\times n\times m}$ where $\bG(i, j, k)$ denotes the $(i, j)$th 
entry of 
$G_k$. Let $\bB'$ be the upper-right $r\times (n-r)\times m$ subtensor of $\bG$, 
with 
$B'_k$ being the corresponding corner in $G_k$. $B'_k$'s span the $\cB'$ as 
defined 
above, 
so $A\in \RAut(\cB')\leq \GL(n-r, q)$ if and only if there exists $D=(d_{i,j})\in 
\GL(m, q)$ 
such that $\forall i\in[m]$, $\sum_{j\in[m]}d_{i,j}B'_j=B'_iA$. 
It is
more 
convenient that we flip $\bB'$ which is of size $r\times (n-r)\times m$ to the 
$\bB$ 
which is of size $(n-r)\times m\times r$ (Figure~\ref{fig:flip}).
Slicing $\bB$ along the third index, 
we obtain an $r$-tuple of $(n-r)\times m$ 
matrices $(B_1, \dots, B_r)$ (Figure~$\ref{fig:slice}$). 
\begin{figure}[!ht]
\centering
\includegraphics[height=3.3cm]{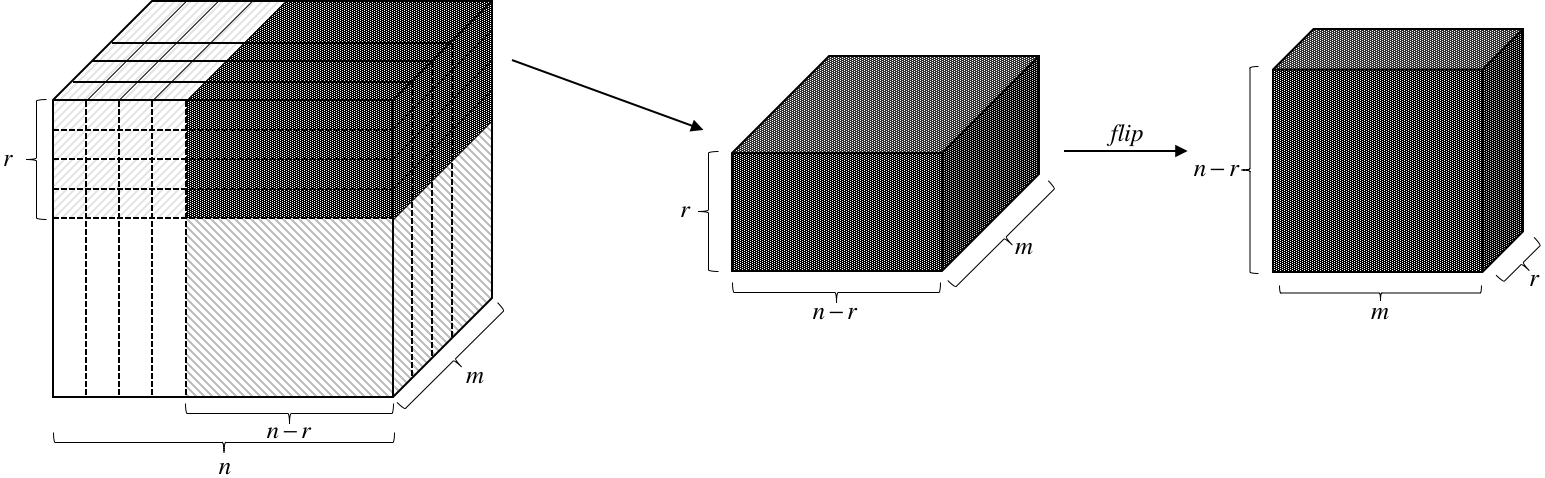}

\caption{The $3$-tensor $\bG$, and flipping $\bB'$ to get $\bB$.}
\label{fig:flip}
\end{figure}

\begin{figure}[!ht]
\centering
\includegraphics[height=3.3cm]{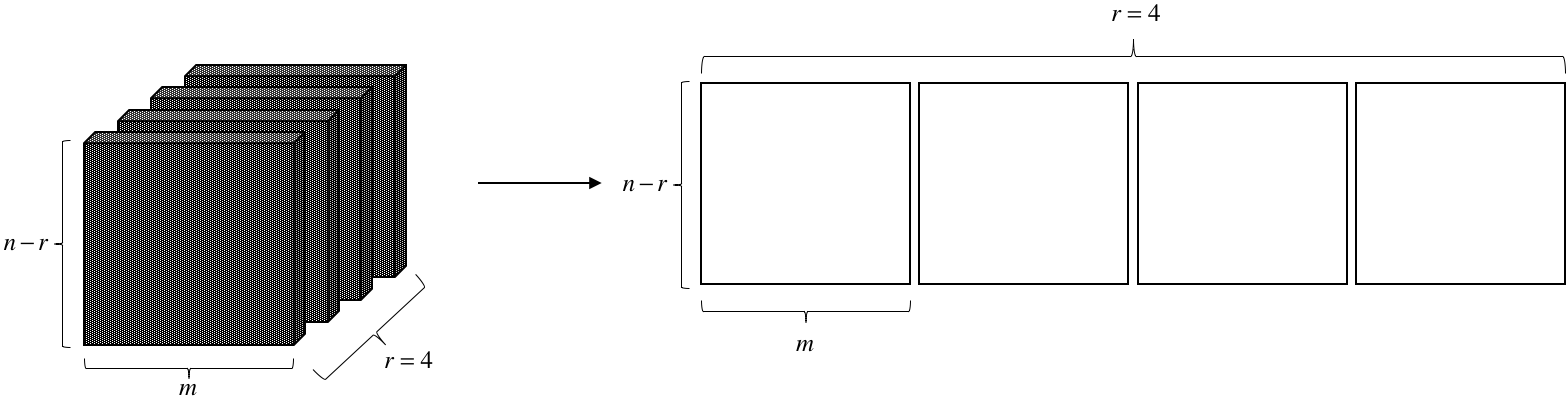}
\caption{Slicing $\bB$.}
\label{fig:slice}
\end{figure}

Define the set of equivalences of $\bB$ as 
$\Aut(\bB):=\{(A, 
D)\in\GL(n-r, q)\times \GL(m, q) : \forall i\in[r], AB_iD^{-1}=B_i\}$.
Note that $\RAut(\cB')$ is 
the projection of $\Aut(\bB)$ to the first component. 
Now define the adjoint 
algebra of $\bB$ as $\Adj(\bB):=\{ (A, D)\in M(n-r, q)\oplus M(m, q) : \forall 
i\in[r], 
AB_i=B_iD\}$. $(A, D)\in M(n-r, q)\oplus M(m, q)$ is called invertible, if both 
$A$ and 
$D$ are invertible. Clearly, $\Aut(\bB)$ consists of the invertible elements in 
$\Adj(\bB)$. 
When 
$r=4$, $m=n-r=n-4$, it can be shown that the adjoint algebra of $4$ random 
matrices in $M(m, q)$ is of size $q^{O(m)}$ with probability $1-1/q^{\Omega(m)}$. 
The key to prove this statement is the stable notion from geometric invariant 
theory \cite{MFK94} in the context of the 
left-right 
action of $\GL(m, q)\times \GL(m, q)$ on matrix tuples $M(m, q)^r$. In this 
context, a matrix tuple 
$(B_1, 
\dots, B_r)\in M(m, q)^r$ is stable, if for every nontrivial subspace $U\leq 
\F_q^n$, $\dim(\langle \cup_{i\in[r]} B_i(U)\rangle)>\dim(U)$. 
An 
upper bound on 
$|\Adj(\bB)|$ can be obtained by analysing this notion using some classical  
algebraic results and elementary probability calculations. 
The good property we impose on $\cG$ is then that the corresponding 
$|\Adj(\bB)|\leq q^{O(m)}$. It can be verified that this property does not depend 
on the choices of bases of $\cG$. 
There is one subtle point though: the analysis on 
$\Adj(\bB)$ is done for $4$ random matrices but we want an analysis for $\cG$ in 
the linear 
algebraic Erd\H{o}s-R\'enyi model. This can be fixed by defining 
a so-called naive model and analysing the relation between the naive model and 
the \LinER model (Section~\ref{sec:naive}).

Now that we have achieved our first goal, namely defining a good property 
satisfied by 
most $\cG$'s, let us see how this property enables an algorithm for such $\cG$'s. 
For an arbitrary $\cH\leq \Lambda(n, q)$, at a multiplicative cost of $q^{O(n)}$ 
(recall thar $r=4$) we can enumerate all $r$-individualisations. Consider a fixed 
one, say $\F_q^n=L\oplus R$ with an ordered basis $(v_1,\dots,v_r)$ of $L$.
Analogous to the above, we can construct $\bC'$, flip to get 
$\bC$, and slice $\bC$ into $r$ $m\times m$ matrices $(C_1, \dots, C_r)$. 
The task then becomes to compute $\Adj(\bB, \bC):=\{(A, D)\in M(n-r, q)\oplus M(m, 
q) : \forall 
i\in[r], AB_i=C_iD\}$. Viewing $A$ and $D$ as variable matrices, $AB_i=C_iD$ are 
linear equations on $A$ and $D$, so the solution set can be computed efficiently. 
As 
$|\Adj(\bB)|\leq q^{O(m)}$, for $\Adj(\bB,\bC)$ to contain an invertible element, 
it must be that $|\Adj(\bB, \bC)|=|\Adj(\bB)|\leq q^{O(m)}$.
In this case 
all elements in $\Adj(\bB, \bC)$ can be enumerated in time $q^{O(m)}=q^{O(n)}$. 
For each element $(A, D)\in \Adj(\bB, \bC)$, test 
whether it is invertible, and if so, test whether the $A$ in that solution induces 
an isometry together with the individualisation. This completes a high-level 
description of the algorithm. In 
particular, this implies that if $\cG$ satisfies this property, then 
$|\Aut(\cG)|\leq q^{O(n)}$. A detailed presentation is in 
Section~\ref{sec:main_algo}, which have some minor differences with the outline 
here, 
as we want to reduce some technical details. 

\section{Preliminaries}

We collect some notation used in this paper. $q$ is reserved 
for prime powers, and $p$ for primes. For $n\in\N$, $[n]:=\{1, \dots, 
n\}$. $\F_q$ denotes the field of size $q$. $\zerovec$ denotes the zero vector or 
the zero vector space. For $i\in[n]$, $e_i$ denotes the $i$th standard basis 
vector of $\F_q^n$. For a 
vector space $V$ and $S\subseteq V$, we use $\langle S\rangle$ to denote the 
linear span of $S$ in $V$. $M(s\times t, q)$ denotes the linear 
space of matrices of size $s\times 
t$ over $\F_q$, and $M(s,q):=M(s\times s, q)$. $I_s$ denotes the $s\times s$ 
identity matrix. For $A\in M(s\times t, q)$, 
$\trans{A}$ denotes the transpose of $A$. 
$\GL(n, q)$ is the general 
linear group consisting of $n\times n$ invertible matrices over $\F_q$.
$\Lambda(n, q)$ is the linear space of alternating matrices of size $n\times 
n$ over $\F_q$. We use $\gbinom{\ }{\ }{q}$ for the 
Gaussian binomial coefficient with base $q$, and $\binom{\ }{\ }$ for the 
ordinary binomial coefficient. For $N\in \N$ and $m\in [N]$, $\gbinom{N}{m}{q}$ 
counts the number of 
dimension-$m$ subspaces in $\F_q^N$. 

By a random vector in $\F_q^N$, we mean a vector of 
length $N$ where each entry is chosen independently and uniformly random from 
$\F_q$. By a random matrix in $M(s\times t, q)$, we mean a matrix of size $s\times 
t$ where each entry is chosen independently and uniformly random from $\F_q$. By a 
random alternating matrix in $\Lambda(n, q)$, we mean an alternating matrix of size $n$ where 
each 
entry in the strictly upper triangular part is chosen independently and uniformly 
random from $\F_q$. Then the diagonal entries are set to $0$, and the lower 
triangular entries are set in accordance with the corresponding upper triangular 
ones.

\begin{fact}\label{fact:prel}
Let $N\in\N$ and $m\in \N$ such that $0\leq m\leq N$. 
\begin{enumerate}
\item For a fixed subspace $U$ in $\F_q^N$ of dimension $m$, the number of complements of $U$ in $\F_q^N$ is $q^{m(N-m)}$; 
\item A random matrix $A\in M(N\times m, q)$ is of rank $m$ with 
probability $\geq 
1-m/q^{N-m+1}$;
\item A random matrix $A\in M(N\times m, q)$ is of rank $m$ with 
probability $> 1/4$.
\end{enumerate}
\end{fact}
\begin{proof}
(1) is well-known. For (2), observe that
\begin{eqnarray*}
\Pr[\rk(A)=m | A\in M(N\times m, q)] & = & 
(1-1/q^N)(1-1/q^{N-1})\dots(1-1/q^{N-m+1}) \\
& \geq & 1-(1/q^N+1/q^{N-1}+\dots+1/q^{N-m+1}) \\ 
& \geq & 1-m/q^{N-m+1}. 
\end{eqnarray*} 
For (3), this is because $\frac{1}{2}\cdot \frac{3}{4}\cdot \frac{7}{8}\cdot 
\dots\approx 0.288788> 1/4$. 
\end{proof}

\section{Matrix tuples and matrix spaces}\label{sec:matrix}

An $r$-matrix tuple of size $s\times t$ over $\F_q$ is an element in $M(s\times 
t, 
q)^r$. An $r$-matrix space of size $s\times t$ over $\F_q$ is a dimension-$r$ 
subspace in $M(s\times t, q)$. An 
$m$-alternating (matrix) tuple of size $n$ over $\F_q$ is an element from 
$\Lambda(n, q)^m$. An 
$m$-alternating (matrix) space of size $n$ over $\F_q$ is a dimension-$m$ 
subspace in $\Lambda(n, q)$.
In the rest of 
this article we let 
$N_n=\binom{n}{2}=\dim(\Lambda(n, q))$, or just $N$ if $n$ is obvious from the 
context. 

We shall use $\cG$, $\cH$, \dots, to denote alternating spaces, and $\bG$, $\bH$, 
\dots, to denote alternating tuples. $\cB$, $\cC$, \dots, are for (not necessarily 
alternating nor square) matrix spaces, and $\bB$, $\bC$ for (not necessarily 
alternating nor square) matrix tuples. We say that a matrix tuple $\bB$ represents 
a 
matrix space $\cB$, if the matrices in $\bB$ form a spanning set (not necessarily 
a basis) of $\cB$. Given 
$A\in M(s, q)$, $D\in M(t, q)$, and $\bB=(B_1, \dots, B_r)\in M(s\times t, q)^r$, 
$A\bB D$ is the tuple $(AB_1D, \dots, AB_rD)$. For $Z=(z_{i,j})\in M(r, q)$, 
$\bB^Z=(\sum_{i\in[r]}z_{1,i}B_i, \sum_{i\in[r]}z_{2,i}B_i, \dots, 
\sum_{i\in[r]}z_{r,i}B_i)$.

Two alternating tuples $\bG=(G_1, \dots, G_m)$ and $\bH=(H_1, \dots, H_m)$ in 
$\Lambda(n, q)^m$ are \emph{isometric}, if there exists $A\in\GL(n, q)$, 
$\trans{A}\bG 
A=\bH$. 
Two alternating spaces $\cG$ and $\cH$ in 
$\Lambda(n, q)$ are \emph{isometric}, if there exists $A\in\GL(n, q)$, such that 
$\trans{A}\cG 
A=\cH$ (equal as subspaces).
Given alternating tuples $\bG\in \Lambda(n, q)^m$ and 
$\bH\in \Lambda(n, q)^m$ representing $\cG$ and $\cH$ respectively, $\cG$ and 
$\cH$ are isometric, if and only if there exists $Z\in\GL(m, q)$ such that $\bG$ 
and $\bH^Z$ are isometric -- in other words, there
exist $A\in \GL(n, q)$ and $Z\in 
\GL(m, 
q)$, such that $\trans{A}\bG A=\bH^Z$. We use $\Iso(\cG,\cH)\subseteq \GL(n, q)$ 
to 
denote 
the set of isometries between $\cG$ and $\cH$. When $\cG=\cH$, the 
isometries between $\cG$ and $\cG$ are also called autometries. The set of all 
autometries forms a matrix group, and let $\Aut(\cG)=\Iso(\cG, \cG)\leq \GL(n, 
q)$. 
$\Iso(\cG,\cH)$ is either empty or a right coset w.r.t. $\Aut(\cG)$. Analogously, 
we can define the 
corresponding concepts for tuples $\Iso(\bG, \bH)$ and $\Aut(\bG)$.
\footnote{We explain our 
choices of the names ``isometry'' and ``autometry''. In 
\cite{Wil09}, for two alternating bilinear maps $b, c:\F_q^n\times 
\F_q^n\to\F_q^m$, an isometry between $b$ and $c$ is $A\in\GL(n, q)$ such that 
$b(A(v_1), A(v_2))=c(v_1,v_2)$ for every $v_1, v_2\in \F_q^n$. A pseudo-isometry 
between $b$ and $c$ is $(A, D)\in \GL(n, q)\times \GL(m, q)$, such that $b(A(v_1), 
A(v_2))=D(c(v_1, v_2))$. The isometry group of $b$ consists of those $A\in\GL(n, 
q)$ preserving $b$ as above, and the pseudo-isometry group of $b$ can also be 
defined naturally. Representing $b$ and $c$ by two alternating matrix tuples, we 
see that 
the isometry (resp. self-isometry) concept there is the same as our isometry 
(resp. autometry) concept for tuples. The pseudo-isometry (resp. 
self-pseudo-isometry) concept corresponds to -- though not exactly the same -- the 
isometry (resp. autometry) concept for spaces. We 
use autometries which seem more convenient and allow for using the notation 
$\Aut$. }

Two matrix tuples $\bB=(B_1, \dots, 
B_r)$ and $\bC=(C_1, 
\dots, C_r)$ in $M(s\times t, q)^r$ are \emph{equivalent}, if there exist $A\in 
\GL(s, 
q)$ and $D\in\GL(t, q)$, such that $A\bB D^{-1}=\bC $. Two 
matrix 
spaces $\cB$ and $\cC$ in $M(s\times t, q)$ are \emph{equivalent}, if there exist 
$A\in\GL(s, q)$ and $D\in\GL(t, q)$, such that $A\cB D^{-1}=\cC$ (equal as 
subspaces). 
By abuse of notation, we use $\Iso(\cB, \cC)\leq \GL(s, q)\times \GL(t, q)$ to 
denote the set of equivalences 
between $\cB$ and $\cC$, and let $\Aut(\cB)=\Iso(\cB, \cB)$. $\Iso(\cB, \cC)$ is 
either empty or a left coset of $\Aut(\cB)$. Similarly we have 
$\Iso(\bB, \bC)$ and $\Aut(\bB)$. A trivial but useful observation is that 
$\Iso(\bB, \bC)$ and $\Aut(\bB)$ are naturally contained in certain subspaces of 
$M(s, q)\oplus M(t, q)$ as follows.\footnote{This linearisation trick allows us to 
decide whether $\bB$ and $\bC$ are equivalent, and compute a generating set of 
$\Aut(\bB)$, by using (sometimes with a little twist) existing algorithms for 
testing module isomorphism 
\cite{CIK97,BL08,IKS10}
and computing the unit group in a matrix algebra \cite{BO08}. On the other hand, 
$\Iso(\bG,\bH)$ and $\Aut(\bG)$ for alternating tuples do not permit such easy 
linearisation. Therefore testing isometry between $\bG$ and $\bH$ \cite{IQ17} and 
computing a generating set for $\Aut(\bG)$ \cite{BW12} requires, besides the 
techniques in \cite{CIK97,BL08,IKS10,BO08}, new ideas, including exploiting the 
$*$-algebra 
structure, the use of which in the context of computing with 
$p$-groups is pioneered by Wilson  \cite{Wil09}.}
Following \cite{Wil09}, we define the \emph{adjoint algebra} of 
$\bB\in 
M(s\times t, q)^r$ as 
$\Adj(\bB)=\{(A, D)\in M(s, q)\oplus M(t, q) : A\bB=\bB D
\}$. This is a classical concept, and is recently 
studied in the context 
of $p$-group isomorphism testing by Wilson et al. 
\cite{Wil09,LW12,BW12,BMW15}. We 
further define the \emph{adjoint space} between $\bB$ and $\bC$ in $M(s\times t, 
q)^r$
as $\Adj(\bB, \bC)=\{(A, D)\in M(s, q)\oplus M(t, q) : A\bB = \bC 
D\}$.
$(A, D)\in 
M(s, q)\oplus M(t, q)$ is called invertible if both $A$ and $D$ are invertible. 
Then $\Aut(\bB)$ (resp. $\Iso(\bB, \bC)$) consists of invertible elements in 
$\Adj(\bB)$ (resp. $\Adj(\bB, \bC)$). An easy observation is that if $\bB$ and 
$\bC$ are isometric, then an isometry defines a bijection between $\Adj(\bB, \bC)$ 
and $\Adj(\bB)$.

Given $\bB=(B_1, \dots, B_r)\in M(s\times t, q)^r$, let $\im(\bB)=\langle 
\cup_{i\in[m]} \im(B_i)\rangle$ and $\ker(\bB)= \cap_{i\in[m]}\ker(B_i)$. $\bB$ is 
image-nondegenerate 
(resp. kernel-nondegenerate), 
if $\im(\bB)=\F_q^s$ (resp. 
$\ker(\bB)=\zerovec$). 
Note that if $\bB$ is an alternating tuple in $\Lambda(n, q)^m$, then $\bB$ is 
image-nondegenerate if 
and only if it is kernel-nondegenerate, as $\im(\bB)$ and $\ker(\bB)$ are 
orthogonal to each other 
w.r.t. the standard bilinear form on $\F_q^n$.
$\bB$ is \emph{nondegenerate} if it is both image-nondegenerate and 
kernel-nondegenerate.
If $\bB$ is image-nondegenerate (resp. kernel-nondegenerate), 
then the projection of 
$\Adj(\bB)$ to the first (resp. second) component along the second (resp. the 
first) component is injective. 

For a matrix tuple $\bB=(B_1, \dots, B_r)\in M(s\times t, q)^r$ and a subspace 
$U\leq \F_q^t$, the image of $U$ under $\bB$ is $\bB(U):=\langle 
\cup_{i\in[m]}B_i(U)\rangle$. It is easy to verify that, $(A\bB)(U)=A(\bB(U))$, 
and $(\bB D)(U)=(\bB(D(U)))$. $U\leq \F_q^t$ is trivial if $U=\zerovec$ or 
$U=\F_q^t$.
\begin{definition}\label{def:stable}
$\bB\in M(s\times t, q)^r$  is \emph{stable}, if $\bB$ is 
nondegenerate, and for 
every nontrivial subspace $U\leq \F_q^t$, $\dim(\bB(U))/\dim(U)>s/t$.
\end{definition} 
\begin{remark}\label{remove_kernel degenerate}
In Definition~\ref{def:stable}, we can replace nondegenerate with 
image-nondegenerate, as the second condition already implies kernel-nondegenerate.
\end{remark}

Lemma~\ref{lem:stable}, \ref{lem:fd_size} and Claim~\ref{claim:trans} are 
classical and certainly known to experts. However for completeness we include 
proofs which may be difficult to 
extract from the literature.

\begin{lemma}\label{lem:stable}
If $\bB$ is stable, then any nonzero $(A, D)\in \Adj(\bB)$ is invertible.
\end{lemma}
\begin{proof}
Take any $(A, D)\in \Adj(\bB)$. If $D=0$, then $A\bB=\bB D=0$, and by the 
image nondegeneracy of $\bB$, $A$ has to be $0$. 

Suppose now that $D$ is not invertible nor $0$, so $\ker(D)$ is not $\zerovec$ nor 
$\F_q^t$. By 
$A\bB(\ker(D))=\bB 
D(\ker(D))=0$, $A(\bB(\ker(D)))=0$, which gives $\ker(A)\geq  \bB(\ker(D))$. 
As $\bB$ is stable, we have $\dim(\bB(\ker(D)))>(s/t)\dim(\ker(D))$, so 
$\dim(\ker(A))>(s/t)\dim(\ker(D))$. On the other hand, $A\bB(\F_q^t)=\bB 
D(\F_q^t)$. Again, by the image nondegeneracy of $\bB$, $\bB(\F_q^t)=\F_q^s$, so 
$A\bB(\F_q^t)=\im(A)$, and we see that $\im(A)=\bB(\im(D))$. As $\bB$ is stable, 
$\dim(\im(A))>(s/t)\dim(\im(D))$. It follows that 
$s=\dim(\im(A))+\dim(\ker(A))>(s/t)(\dim(\im(D))+\dim(\ker(D)))=(s/t)\cdot t=s$. 
This is a contradiction, so $D$ has to be invertible. 

If $D$ is invertible, then $\bB D$ is image-nondegenerate, so $A$ has to be 
invertible, as otherwise $A\bB$ would not be image-nondegenerate. 
\end{proof}
\begin{remark}\label{rem:stable}
We present some background information on the stable concept and 
Lemma~\ref{lem:stable}, for readers who have not encountered these before. Briefly 
speaking, the stable concept is a correspondence of the concept of simple as in 
representation theory of associative algebras, and Lemma~\ref{lem:stable} 
is an analogue of the Schur's lemma there. Both the stable concept here and the 
simple concept are special cases of the stable concept in geometric invariant 
theory \cite{MFK94,Kin94}, specialised to the left-right action of $\GL(s, 
q)\times \GL(t, q)$ on $M(s\times t, q)^r$, and the conjugation action of $\GL(s, 
q)$ on $M(s, q)^r$, respectively.

Specifically, consider a tuple of square matrices $\bB\in M(s, q)^r$, which can be 
understood as a representation of an associative algebra with $r$ generators. This 
representation is simple if and only if it does not have a non-trivial invariant 
subspace, 
that is $U\leq \F_q^s$, such that $\bB(U)\leq U$. This amounts to say that there 
does not exist $A\in \GL(s, q)$ such that every $B$ in $A\bB A^{-1}$ is in the 
form $\begin{bmatrix} B_1 & B_2 \\ 0 & B_3\end{bmatrix}$ where $B_i\in M(s', q)$, 
$1\leq s'\leq s-1$. On the other hand, the 
stable concept can be rephrased as the following. $\bB\in M(s\times t, q)^r$ is 
stable, if there do not exist $A\in\GL(s, q)$ and $D\in\GL(t, q)$ such that every 
$B\in A\bB D^{-1}$ is of the form $\begin{bmatrix} B_1 & B_2 \\ 0 & 
B_3\end{bmatrix}$ 
where $B_1$ is of size $s'\times t'$, $1\leq t'\leq t-1$, such that $s'/t'\leq 
s/t$. 

Lemma~\ref{lem:stable} can be understood as an analogue of Schur's lemma, which 
states that if $\bB\in M(s, q)^r$ is 
simple then a nonzero homomorphism $A\in M(s, q)$ of $\bB$ (e.g. $A\bB 
A^{-1}=\bB$) has to 
be invertible. 
\end{remark}

The proof of the following classical result was communicated to us by G. Ivanyos. 
\begin{lemma}\label{lem:fd_size}
Let $\mathfrak{A}\subseteq M(n, q)$ be a field containing $\lambda I_n$, 
$\lambda\in \F_q$. Then $|\mathfrak{A}|\leq q^n$. 
\end{lemma}
\begin{proof}
$\mathfrak{A}$ is an extension field of $\F_q$, and suppose its extension degree 
is $d$. Then $\F_q^n$ is an $\mathfrak{A}$-module, or in other words, a vector 
space over $\mathfrak{A}$. So $\F_q^n\cong \mathfrak{A}^m$ as 
vector spaces over $\mathfrak{A}$ for some $m\in \N$. Considering them as $\F_q$ 
vector spaces, we 
have $n=md$ so $d$ divides $n$. It follows that $|\mathfrak{A}|=q^d\leq q^n$.
\end{proof}

By Lemma~\ref{lem:stable} and~\ref{lem:fd_size}, we have the following.
\begin{proposition}\label{prop:stable_size}
If $\bB\leq M(s\times t, q)^r$ is stable, then $|\Adj(\bB)|\leq q^{s}$.
\end{proposition}
\begin{proof}
As $\bB$ is stable, it is nondegenerate, so the projection of $\Adj(\bB)\leq M(s, 
q)\oplus M(t, q)$ to $M(s, q)$ (naturally embedded in $M(s, q)\oplus M(t, q)$) 
along $M(t, q)$ is injective. 
By 
Lemma~\ref{lem:stable}, the image of the projection is a finite division algebra 
over $\F_q$ containing 
$\lambda I$. So by Wedderburn's little theorem, it is a field. By 
Lemma~\ref{lem:fd_size}, the result follows.
\end{proof}

Let us also mention an easy property about stable.
\begin{claim}\label{claim:trans}
Given $\bB=(B_1, \dots, B_r)\in M(s\times t, q)^r$, let $\trans{\bB}=(\trans{B}_1, 
\dots, 
\trans{B}_r)\in M(t\times s, q)^r$. Then $\bB$ is stable if and only if 
$\trans{\bB}$ is 
stable.
\end{claim}
\begin{proof}
First we consider the nondegenerate part. If $u\in \F_q^s$ satisfies 
$\bB(u)=\zerovec$, 
then it is 
easy to verify that $\trans{\bB}(\F_q^t)$ is contained in the hyperplane defined 
by $u$, 
e.g. $\trans{u}(\trans{\bB}(\F_q^s))=0$. If $\bB(\F_q^t)\neq \F_q^s$, then there 
exists some 
$u\in \F_q^s$ such that $\trans{u}(\bB(\F_q^t))=0$, so $u\in \ker(\trans{\bB})$. 
Therefore $\bB$ is nondegenerate if and only if $\trans{\bB}$ is nondegenerate.

In the following we assume that $\bB$ is nondegenerate, and check nontrivial 
subspaces to  show that $\bB$ is not stable if and only 
if $\trans{\bB}$ is not stable. This can be seen easily from the discussion in 
Remark~\ref{rem:stable}. $\bB$ is not stable, then there exist $A\in\GL(r, q)$ and 
$D\in\GL(t, q)$ 
such that every 
$B\in A\bB D^{-1}$ is of the form $\begin{bmatrix} B_1 & B_2 \\ 0 & 
B_3\end{bmatrix}$ 
where $B_1$ is of size $s'\times t'$, $1\leq t'\leq t-1$, such that $s'/t'\leq 
s/t$. Note that $s'>0$ as otherwise $\bB$ is degenerate, so $1\leq s'\leq 
\frac{s}{t}\cdot t'<s$.
Now consider $D^{-\mathrm{t}}\trans{\bB}\trans{A}$, the elements in which is of 
the 
form $\begin{bmatrix} \trans{B}_1 & 0\\ \trans{B}_2 & 
\trans{B}_3\end{bmatrix}$. Note that $\trans{B}_3$ is of size $(t-t')\times 
(s-s')$ where 
$1\leq s-s'\leq s-1$, $1\leq t-t'\leq t-1$, and $(t-t')/(s-s')\leq t/s$ (by 
$s'/t'\leq s/t$). It follows that $D^{-\mathrm{t}}\trans{\bB}\trans{A}$ is not 
stable, so $\trans{\bB}$ is not stable. This concludes the proof. 
\end{proof}

%


\section{Random alternating matrix spaces}\label{sec:naive}

For $n\in \N$, $N=\binom{n}{2}$. Recall the definition of the linear algebraic 
Erd\H{o}s-R\'enyi model, $\LinER(n, m)$ in Model~\ref{model:linER}. It turns out 
for our purpose, we can work with the following model.
%

\begin{model}[Naive models for matrix tuples and matrix spaces]
The naive model for alternating tuples, $\NaiT(n, m, q)$, is the probability 
distribution over the set of all $m$-tuples of $n\times n$ alternating tuples, 
where each tuple is endowed with probability $1/q^{N m}$.

The naive model for alternating spaces, $\NaiS(n, m, q)$, is the probability 
distribution over the set of alternating spaces in $\Lambda(n, q)$ of 
dimension $\leq m$, where the probability at some $\cG\leq \Lambda(n, q)$ of 
dimension $0\leq d\leq m$ equals the number of $m$-tuples of $n\times n$ 
alternating tuples that represent $\cG$, divided by $q^{N m}$. 
\end{model}


While we aim at analysing the algorithm in the $\LinER$ model, we will ultimately 
work with 
the naive model due to its simplicity, as 
it is just an
$m$-tuple of random $n\times n$ alternating matrices. The naive 
model for alternating spaces, $\NaiS$, then is obtained by taking the linear spans 
of such tuples. The following observation will be useful.
\begin{observation}\label{obs:naive}
Every $m$-alternating space has $(q^m-1)(q^m-q)\dots(q^m-q^{m-1})$ $m$-alternating 
tuples representing it. 
\end{observation}


We now justify that working with the naive model suffices for 
the analysis even in the linear algebraic Erd\H{o}s-R\'enyi model. 
Consider the following setting. Suppose we have $E(n, m, q)$, a 
property of 
dimension-$m$ alternating spaces in 
$\Lambda(n, q)$, and wish to show that $E(n, m, q)$ holds with high probability in 
$\LinER(n, m, q)$. 
$E(n, m, q)$ naturally induces $E'(n, m, q)$, a property of alternating 
tuples in $\Lambda(n, q)^m$ that span dimension-$m$ alternating spaces. 
It is usually the case that there exists 
a property $F(n, m, q)$ of all $m$-alternating tuples in $\Lambda(n, q)^m$, so 
that 
$F(n, m, q)$ and $E'(n, m, q)$ coincide when restricting to those alternating 
tuples spanning dimension-$m$ matrix spaces. If we could 
prove 
that $F(n, m, q)$ holds with high probability, then since a nontrivial fraction 
of $m$-tuples do span 
dimension-$m$ spaces, 
we would get that $E(n, 
m, q)$ holds with high probability 
as well. The following proposition summarises and makes precise the above 
discussion. 

\begin{proposition}\label{prop:models}
Let $E(n, m, q)$ and $F(n, m, q)$ be as above. 
Suppose in $\NaiT(n, m, q)$, $F(n, m, q)$ happens with probability $\geq 1-f(n, m, 
q)$ where $0\leq f(n, m, q)< 1$. Then in $\LinER(n, m, q)$, $E(n, m, q)$ 
happens with probability $> 1- 4\cdot f(n, m, q)$.
\end{proposition}
\begin{proof}
The number of tuples for which $F(n, m, q)$ fails is no larger than $f(n,m,q)\cdot q^{N m}$. 
Clearly the bad situation for $E'(n, m, q)$ is when each of them spans an 
$m$-alternating space, so we focus on this case. Recall that $E'(n, m, 
q)$ is induced from a property of $m$-alternating spaces. That is, if two tuples 
span the same $m$-alternating space, then either both of them satisfy $E'(n, m, 
q)$, or neither of them satisfies $E'(n, m, q)$. By Observation~\ref{obs:naive}, 
the number of 
$m$-alternating spaces for which $E(n, m, q)$ fails is $\leq f(n, m, q)\cdot 
\frac{q^{N m}}{(q^m-1)(q^m-q)\dots(q^m-q^{m-1})}$. The fraction of $m$-alternating 
spaces for which $E(n, m, q)$ fails is then $\leq f(n, m, q)\cdot \frac{q^{N 
m}}{(q^N-1)(q^N-q)\dots (q^N-q^{N-m+1})}< 4\cdot f(n, m, q)$ where $4$ comes 
from 
Fact~\ref{fact:prel} (3).
\end{proof}

\subsection{Random matrix spaces}\label{subsec:random_bip}

For $s, t, r \in \Z^+$, we can define the Erd\H{o}s-R\'enyi model for 
bipartite graphs on the vertex set $[s]\times [t]$ with edge set size $r$ by 
taking every subset of $[s]\times [t]$ of size $r$ with probability 
$\binom{st}{r}$. Analogously we can define the following in the matrix space 
setting.
\begin{enumerate}
\item The bipartite linear algebraic Erd\H{o}s-R\'enyi model $\BipLinER(s\times t, 
r, q)$: each $r$-matrix space in $M(s\times t, q)$ is chosen with probability 
$1/\gbinom{st}{r}{q}$.
\item The bipartite naive model $\BipNaiT(s\times t, r, q)$
for matrix tuples: each 
$r$-matrix tuple in $M(s\times t, q)^r$ is chosen with probability $1/q^{str}$.
\item The bipartite naive model $\BipNaiS(s\times t, r, q)$ for matrix spaces: each 
matrix space $\cB$ of dimension $d$, $0\leq d\leq r$ in $M(s\times t, q)$, is 
chosen with probability $a/q^{str}$ where $a$ 
is the number of $r$-matrix 
tuples representing $\cB$.
\end{enumerate}

\section{The main algorithm}\label{sec:main_algo}


We will first define the property $F(n, m, q, r)$ for the average-case analysis in 
Section~\ref{subsec:property}. To lower bound the probability of $F(n, m, q, r)$ 
we will actually work 
with a stronger property $F'(n, m, q, r)$ in Section~\ref{subsubsec:prob}. Given 
this property we describe and 
analyse the main algorithm in Section~\ref{subsec:main_algo}. It should be noted 
that the algorithm here differs slightly from the outline from 
Section~\ref{subsec:outline}, as there we wanted to reduce some technical details. 

\subsection{Some properties of alternating spaces and alternating 
tuples}\label{subsec:property}

An $m$-alternating space $\cG\leq \Lambda(n, q)$ 
induces $\cB'=\{[e_1, \dots, e_r]^t G [e_{r+1}, \dots, e_n] : G\in \cG\}$ which is 
a matrix space in $M(r\times (n-r), q)$ of dimension no more than $m$. Define 
$\RAut(\cB'):=\{A\in \GL(n-r, q) : \cB' A=\cB'\}$. An element in $\RAut(\cB')$ is 
called a right-side equivalence of $\cB'$.
%
%

\begin{definition}
$E'(n, m, q, r)$ is a property of $m$-alternating spaces in $\Lambda(n, q)$, 
defined as 
follows. Given an $m$-alternating space $\cG$ in $\Lambda(n, q)$, let $\cB'$ be 
the matrix space in $M(r\times (n-r), q)$ defined as above. $\cG$  belongs to 
$E'(n,m,q, r)$, if and only if $|\RAut(\cB')|\leq q^{n-r}$.
\end{definition}

Right-side equivalence is a useful concept that leads to our algorithm (as seen 
in Section~\ref{subsec:ind}), but 
what we actually need is the following linearisation of 
$\RAut(\cB')$. 

\begin{definition}
$E(n, m, q, r)$ is a property of $m$-alternating spaces in $\Lambda(n, q)$, 
defined as 
follows. Given an $m$-alternating space $\cG$ in $\Lambda(n, q)$, let $\cB'$ be 
the matrix space in $M(r\times (n-r), q)$ defined as above. $\cG$  belongs to 
$E(n,m,q, r)$, if and only if $|\{A\in M(n-r, q) : \cB'A\leq \cB'\}|\leq 
q^{n-r}$.
\end{definition}

We define a property $F(n, m, q, r)$ for alternating tuples that corresponds to 
$E(n, m, q, r)$. Given $\bG=(G_1, \dots, G_m)\in \Lambda(n, q)^m$, we can 
construct a matrix tuple 
$\bB'=([e_1, \dots, 
e_r]^t 
B_1 [e_{r+1}, \dots, e_n], \dots,$ $[e_1, \dots, e_r]^t B_m [e_{r+1}, \dots, 
e_n])$ in $M(r\times (n-r), q)^m$. 
\begin{definition}\label{def:property_F}
$F(n, m, q, r)$ is a property of $m$-alternating tuples in $\Lambda(n, q)^m$, 
defined as 
follows. Given an $m$-alternating tuple $\bG$ in $\Lambda(n, q)^m$, let $\bB'$ 
be the $m$-matrix tuple in $M(r\times (n-r), q)^m$ defined as above. $\bG$  
belongs to 
$E(n,m,q, r)$, if and only if $|\{A\in M(n-r, q):~\exists D\in M(m, q),~ 
\bB'A=\bB'^D\}|\leq 
q^{n-r}$.
\end{definition}

It is not hard to see that $F(n, m, q, r)$ is a proper extension of $E(n, m, q, 
r)$.
\begin{proposition}
Suppose $\bG\in \Lambda(n, q)^m$ represents an $m$-alternating space 
$\cG\leq \Lambda(n, q)$. Then $\bG$ is in $F(n, m, q, r)$ if and only if $\cG$ is 
in $E(n, m, q, r)$. 
\end{proposition}
\begin{proof}
Let $\cB'$ and $\bB'$ be the matrix space and matrix tuple defined as above for 
$\cG$ and $\bG$, respectively. Clearly $\bB'$ represents $\cB'$, so $\bB' A$ 
represents $\cB' A$. Finally note that $\bB' A=\bB'^D$ for some $D\in M(n, q)$ if 
and only if the linear span of $\bB' A$ is contained in the linear span of 
$\bB'$, that is $\cB'A\leq \cB'$.
\end{proof}

Instead of working with $\bB'$ and $\{A\in \GL(n-r, q): \exists D\in \GL(m, q), 
\bB' A=\bB'^D \}$, it is more convenient to flip $\bB'$, an $m$-matrix tuple 
of size $r\times (n-r)$, to get $\bB$, an $r$-matrix tuple of size $(n-r)\times 
m$. Then $\{A\in M(n-r, q): \exists D\in M(m, q), \bB' A=\bB'^D\}=\{A\in 
M(n-r, q): \exists D\in M(m, q), A\bB = \bB D\}$. The latter is closely related to 
the 
adjoint algebra concept for matrix tuples as defined in Section~\ref{sec:matrix}.
Recall that $\Adj(\bB)=\{(A, D)\in M(n-r, q)\oplus M(m, q) : A\bB = \bB D\}$. Let 
$\pi_1: M(n-r, q)\oplus M(m, q)\to M(n-r, q)$ be the projection to the first 
component along the second. 
$\{A\in M(n-r, q): \exists D\in M(m, q), A\bB = \bB D\}$ is then just 
$\pi_1(\Adj(\bB))$. So Definition~\ref{def:property_F} is equivalent to the 
following. 

\medskip
\noindent{\bf Definition~\ref{def:property_F}, alternative formulation.} $F(n, m, 
q, 
r)$ is a property of $m$-alternating tuples in $\Lambda(n, q)^m$, 
defined as 
follows. Given an $m$-alternating tuple $\bG$ in $\Lambda(n, q)^m$, let $\bB$ 
be the $m$-matrix tuple in $M((n-r)\times m, q)^r$ defined as above. $\bG$  
belongs to 
$E(n,m,q, r)$, if and only if $|\pi_1(\Adj(\bB))|\leq q^{n-r}$.
\medskip

Our algorithm will be based on the property $F(n, m, q, r)$. To show that 
$F(n, m, q, r)$ holds with high probability though, we turn to study 
the following stronger property. 

\begin{definition}\label{def:property_F_prime}
$F'(n, m, q, r)$ is a property of $m$-alternating tuples in $\Lambda(n, q)^m$, 
defined 
as follows. Given an $m$-alternating tuple $\bG$ in $\Lambda(n, q)^m$, let $\bB$ 
be the $r$-matrix tuple in $M((n-r)\times m, q)^r$ defined as above. $\bG$ 
 belongs to $F'(n, 
m, 
q, r)$, if and only if $|\Adj(\bB)|\leq q^{n-r}$.
\end{definition}

Clearly $F'(n, m, q, r)$ implies $F(n, m, q, r)$. To show that $F'(n, m, q, r)$ 
holds with high probability, Proposition~\ref{prop:stable_size} immediately 
implies the 
following, which directs us to make use of the stable property.
\begin{proposition}\label{prop:stable}
Let $\bG$ and $\bB$ be defined as above. If $\bB$ is stable, then $\bG\in F'(n, 
m, q, r)$.
\end{proposition}

\subsubsection{Estimating the probability for the property $F'(n, m, q, 
r)$}\label{subsubsec:prob}

We now show that $F'(n, m, q, r)$ holds with high probability, when 
$m=cn$ for some 
constant $c$ with an appropriate choice of $r$ depending on $c$.
The integer $r$ is chosen so  
that $r\geq 4\cdot \frac{n-r}{m}$ if $n-r\geq m$, and $r\geq 4\cdot \frac{m}{n-r}$ 
if $m\geq n-r$. When $n$ is large enough this is always possible. 
For example, if $c\geq 1$, let $r$ be any integer $\geq 5c$, which ensures that 
$r\geq 4\cdot 
\frac{m}{n-r}$ if $n\geq 25c$. If $0<c<1$, let $r$ be an integer $\geq 5/c$. If 
$n-5/c\geq m$, then $r\geq 5\cdot \frac{n-r}{m}\geq 4\cdot \frac{n-r}{m}$. If 
$n-5/c<m$, then $r\geq 5\cdot \frac{n-r}{m}\geq 4\cdot \frac{m}{n-r}$ if $n\geq 
\frac{r}{1-2c/\sqrt{5}}$. 

Let $s=n-r$ and 
$t=m$. 
By Proposition~\ref{prop:stable}, to show $F'(n, m, q, r)$ holds with high 
probability, we can show that for most $\bG$ from 
$\NaiT(n, m, q)$, the 
corresponding 
$\bB$ in $M(s\times t, q)^r$ is stable. A simple observation is that $\NaiT(n, m, 
q)$ induces $\BipNaiT(s\times t, r, q)$ obtained by flipping the upper right 
$s\times t$ corners of 
the alternating matrices (see Figure~\ref{fig:flip}). So we reduce to estimate the 
probability of an $r$-matrix tuple $\bB$ in 
$M(s\times t, q)^r$ being stable in the model $\BipNaiT(s\times t, r, q)$.



By our choice of $r$, we obtain an $r$-matrix tuple $M(s\times t, q)$ 
with $r\geq 4\cdot \frac{\max(s, t)}{\min(s, t)}$. By Claim~\ref{claim:trans}, 
we know $\Pr[\bB \text{ is 
stable in }\BipNaiT(s\times t, r, 
q)]=\Pr[\bC \text{ is stable in }\BipNaiT(t\times s, r, q)]$ via the transpose 
map. So it is enough to consider the case when $s\geq t$. 

\begin{proposition}\label{prop:key}
Give positive integers $s$, $t$, and $r$ such that $s\geq t\geq 16$, 
$\frac{s}{t}=b\geq 1$, and $r\geq 
4\cdot \frac{s}{t}$. Then 
$\bB$ is stable with probability $1-\frac{1}{q^{\Omega(t)}}$ in $\BipNaiT(s\times 
t, r, q)$, where $\Omega(t)$ hides a positive constant depending on $b$.
\end{proposition}
\begin{proof}
We will upper bound the probability of $\bB$ being not stable in $\BipNaiT(s\times t, 
r, 
q)$, which is
\begin{equation*}
P=\Pr[\bB~\text{is degenerate, or }\exists U\leq \F_q^t, U \text{ non-trivial}, 
\frac{{\dim}(\bB(U))}{\dim(U)}\leq \frac{s}{t}].
\end{equation*}
By the union bound, we have:
\begin{equation*} 
P\leq \sum_{U\leq \F_q^{t},1\leq \dim(U)\leq 
t-1}\Pr[\frac{{\dim}(\bB(U))}{\dim(U)}\leq \frac{s}{t} ]+\Pr[\bB~{\rm 
is~degenerate}].
\end{equation*}

\paragraph{About $\bB$ being degenerate.} 
By Remark~\ref{remove_kernel degenerate}, we only need to bound $\Pr[\bB 
\text{ is image-degenerate}]$. Noticing that $\im(\bB)$ is spanned by the columns of $B_i$'s, by 
forming an $s\times rt$ matrix $A=[B_1,B_2,\dots,B_r]$, this 
amounts to upper bound the probability that $\Pr[ \rk(A)<s | A\in M(s\times rt, 
q)]$. As $rt\geq 4bt= 4s$, $$\Pr[\rk(A)=s | A\in M(s\times rt, q)]\geq \Pr[\rk(A)=s 
| A\in M(s\times 4s, q)]\geq 1-s/q^{3s+1},$$ where the last inequality is from 
Fact~\ref{fact:prel} (2). So we have $\Pr[\bB\text{ is image-degenerate}]\leq 
1/q^{\Omega(t)}$ since $s=bt$.

\paragraph{Reduce to work with nontrivial subspaces according to the 
dimension 
$d$.} Now we focus on $\sum_{U\leq \F_q^{t}, 1\leq \dim(U)\leq 
t-1}\Pr[\frac{{\dim}(\bB(U))}{\dim(U)}\leq \frac{s}{t}]$ in the following.
For a nontrivial subspace $U\leq \F_q^t$, let $$B_U=\{\bB\in M(s\times 
t, q)^r:\frac{{\dim}(\bB(U))}{\dim(U)}\leq\frac{s}{t}\}.$$ Consider two subspace 
$U_1,U_2\leq \F_{q}^t$ of the same dimension $1\leq d\leq t-1$. We claim that 
$|B_{U_1}|=|B_{U_2}|$.
Let $X\in\GL(t,q)$ be any invertible matrix such that $X(U_2)=U_1$, and consider 
the map $T_X:M(s\times t, q)^r\to M(s\times t, q)^r$ defined by sending $\bB$ to 
$\bB X$. 
It is easy to verify that $T_X$ is a bijection between $B_{U_1}$ and 
$B_{U_2}$.
The claim then follows and we have $\Pr_\bB[\frac{{\dim}(\bB(U_1))}{\dim(U_1)}\leq 
\frac{s}{t}]=\Pr_\bB[\frac{{\dim}(\bB(U_2))}{\dim(U_2)}\leq \frac{s}{t}]$. So 
setting $U_d=\langle 
e_1,\dots,e_d\rangle$, we have 
\begin{equation*}
\sum_{U\leq \F_q^{t},1\leq \dim(U)\leq t-1}\Pr[\frac{{\dim}(\bB(U))}{\dim(U)}\leq 
\frac{s}{t}]=\sum_{1\leq d\leq 
t-1}\gbinom{t}{d}{q}\Pr[{\dim}(\bB(U_d))\leq 
\frac{s}{t}\cdot d].
\end{equation*}

\paragraph{Upper bound $\gbinom{t}{d}{q}\Pr[{\dim}(\bB(U_d))\leq 
\frac{s}{t}\cdot d]$.}
For $1\leq d\leq t-1$, let $P_d=\Pr[{\dim}(\bB(U_d))\leq \frac{s}{t}\cdot d]$. For any matrix 
$B\in 
M(s\times t, q)$, $B(U_d)$ is spanned by the first $d$ column vectors of $B$. So 
for $\bB=(B_1, \dots, B_r)\in M(s\times t, q)^r$, $\bB(U_d)$ is spanned by the 
first $d$ columns of $B_i$'s. Collect those columns to form a matrix $A\in 
M(s\times rd, q)$, and we have 
\begin{equation}\label{equivalentform}
P_d=\Pr[{\dim}(\bB(U_d))\leq bd]=\Pr[\rk(A)\leq \lfloor bd\rfloor|A \in M(s\times 
rd,q)].
\end{equation}
Note that in the above we substituted $bd$ with $\lfloor bd\rfloor$ as that does 
not change the probability.

Equation \ref{equivalentform} suggests the following upper bound of $P_d$. For $A$ 
to be of rank $\leq \lfloor bd\rfloor$, there must exist $\lfloor bd\rfloor$ 
columns such that other columns are linear combinations of them. So we enumerate 
all subsets of the columns of size $\lfloor bd\rfloor$, fill in these columns arbitrarily, 
and let other columns be linear combinations of them. This shows that 
\begin{equation}\label{upper_bound}
P_d\leq \frac{\binom{rd}{\lfloor bd\rfloor}\cdot q^{s\lfloor bd\rfloor}\cdot 
q^{\lfloor bd\rfloor(rd-\lfloor bd\rfloor)}}{q^{srd}}.
\end{equation}

When $1\leq d\leq t/2$, we have
\begin{equation}\label{1-t/2}
\begin{split}
\gbinom{t}{d}{q}P_d&\leq\frac{\binom{rd}{\lfloor bd\rfloor}\cdot q^{s\lfloor bd\rfloor}\cdot q^{\lfloor bd\rfloor(rd-\lfloor bd\rfloor)}\cdot \gbinom{t}{d}{q}}{q^{srd}}\\
&\leq\frac{q^{rd}\cdot q^{sbd}\cdot q^{bd(rd-bd)}\cdot q^{td}}{q^{srd}}\\
&\leq \frac{1}{q^{(sr-sb-t-r)d-b(r-b)d^2}},
\end{split}
\end{equation}
where in the second inequality, we use $\binom{rd}{\lfloor bd\rfloor}\leq 
2^{rd}\leq q^{rd}$, $\gbinom{t}{d}{q}\leq q^{td}$ and $\lfloor 
bd\rfloor(rd-\lfloor bd\rfloor)\leq bd(rd-bd)$ since $r\geq 4b$.

Let $f(d)=(sr-sb-t-r)d-b(r-b)d^2$. 
It is easy to see that 
$f(d)$ achieves minimum at $d=1$ or $d=t/2$ in the interval 
$1\leq d\leq \frac{t}{2}$. We have 
$f(1)=(br-b^2-1)t+b^2-r-br$
and
$f(\frac{t}{2})=(\frac{1}{4}br-\frac{1}{4}b^2-\frac{1}{2})t^2-\frac{1}{2}rt$.
Since $r\geq 4b$ and $b\geq 1$, $br-b^2-1\geq 3b^2-1>0$ and 
$\frac{1}{4}br-\frac{1}{4}b^2-\frac{1}{2}\geq\frac{3}{4}b^2-\frac{1}{2}>0$. These 
two lower bounds then yield that 
$\gbinom{t}{d}{q}P_d\leq \frac{1}{q^{\Omega(t)}}$	
for $1\leq d\leq t/2$.
\medskip

When $t/2\leq d\leq t-3$, we replace $\gbinom{t}{d}{q}$ by $\gbinom{t}{t-d}{q}$ in Inequality~\ref{1-t/2} and obtain
\begin{equation*} 
\begin{split}
\gbinom{t}{d}{q}P_d
&\leq\frac{1}{q^{(sr-sb+t-r)d-b(r-b)d^2-t^2}},
\end{split}
\end{equation*}

It can be seen easily that the function $g(d)=(sr-sb+t-r)d-b(r-b)d^2-t^2$ achieves 
minimum at either 
$d=t/2$ or $d=t-3$. We have 
$g(\frac{t}{2})=f(\frac{t}{2})=(\frac{1}{4}br-\frac{1}{4}b^2-\frac{1}{2})t^2
-\frac{1}{2}rt$
and
$g(t-3)=(3br-3b^2-r-3)t+3r+9b^2-9br$.
Since $r\geq 4b$ and $b\geq 1$,
$\frac{3}{4}b^2-\frac{1}{2}>0$ and $9b^2-4b-3>0$ when $b\geq 1$. These two lower bounds then yield that 
$\gbinom{t}{d}{q}P_d\leq \frac{1}{q^{\Omega(t)}}$ for $t/2\leq d\leq t-3$.

For $d=t-2$ and $t-1$, we use the method for the nondegenerate part. 
Recall that
$P_{t-2}=\Pr[\rk(A)\leq  b(t-2) | A\in M(s\times r(t-2),  q)]$.
When $t\geq 16$ (i.e. $s\geq 16b$), $r(t-2)\geq 4b(t-2)\geq \lceil\frac{7}{2}s\rceil$. 
Also 
note that $b(t-2)< bt=s$. 
Therefore $P_{t-2}\leq \Pr[\rk(A)< s | A\in M(s\times \lceil\frac{7}{2}s\rceil, 
q)]$, which is $\leq s/q^{\frac{5}{2}s+1}$ by Fact~\ref{fact:prel} (2). Then 
$\gbinom{t}{t-2}{q}P_{t-2}\leq \frac{sq^{2t}}{q^{\frac{5}{2}s+1}}=\frac{bt}{q^{(\frac{5}{2}b-2)t+1}}\leq 
\frac{1}{q^{\Omega(t)}}$. The case when $d=t-1$ is similar, and we can obtain $\gbinom{t}{t-1}{q}P_{t-1}\leq \frac{1}{q^{\Omega(t)}}$ as well. This concludes the proof.
\end{proof}

\subsection{The algorithm}\label{subsec:main_algo}

We now present a detailed description and analysis of the main algorithm and prove 
Theorem~\ref{thm:main}. 

As described in Section~\ref{subsec:outline}, the concept of $r$-individualisation 
is a key technique in the algorithm. Recall that an $r$-individualisation is a 
direct 
sum decomposition $\F_q^n=L\oplus R$ with an ordered basis $(v_1, \dots, v_r)$ of 
$L$. In the algorithm we will need to enumerate all $r$-individualisations, and 
the following proposition realises this. 
\begin{proposition}\label{prop:ind}
There is a deterministic algorithm that lists all $r$-individualisations in 
$\F_q^n$ in time $q^{O(rn)}$. Each individualisation $L\oplus R$ with an ordered 
basis $(v_1, \dots, v_r)$ of $L$ is represented as an invertible matrix $[v_1, 
\dots, v_r, u_1, \dots, u_{n-r}]$ where $(u_1, \dots, u_{n-r})$ is an ordered 
basis of $R$.
\end{proposition}
\begin{proof}
Listing all $r$-tuples of linearly independent vectors can be done easily in time 
$q^{rn}\cdot \poly(n, \log q)$. For a dimension-$r$ $L\leq \F_q^n$ with an ordered 
basis $(v_1, 
\dots, v_r)$, we need to compute all complements of $L$, and represent every 
complement by an 
ordered basis. To do this, we first compute 
one ordered basis of one 
complement of $L$, which can be easily done as this just means to compute a full 
ordered basis starting from a partial order 
basis.
Let this 
ordered basis 
be $(u_1, \dots, u_{n-r})$. Then the spans of the $(n-r)$-tuples $(u_1+w_1, \dots, 
u_{n-r}+w_{n-r})$ 
go over all complements of $L$ when $(w_1, \dots, w_{n-r})$ go over all 
$(n-r)$-tuples 
of vectors from $L$. Add $[v_1, \dots, v_r, u_1+w_1, \dots, 
u_{n-r}+w_{n-r}]$ to the list. The total number of 
iterations, namely $r$-tuples of vectors from $\F_q^n$ and $(n-r)$-tuples of 
vectors from $L\cong \F_q^r$, 
is $q^{2rn-r^2}$. Other steps can be achieved via linear algebra computations. 
This concludes the proof. 
\end{proof}
\begin{remark}\label{rem:ind}
The algorithm in Proposition~\ref{prop:ind} produces a list $T$ of invertible 
matrices 
of size $n$. 
An invertible $A_0=[v_1, \dots, v_n]\in 
\GL(n, q)$, viewed as a change-of-basis matrix, sends $e_i$ to $v_i$ for 
$i\in[r]$, and $\langle e_{r+1}, \dots, e_n\rangle$ to $R=\langle v_{r+1}, \dots, 
v_n\rangle$. Suppose $A_1=[v_1, \dots, v_r, u_1, \dots, 
u_{n-r}]$ is the matrix from $T$ where $\langle u_1, \dots, 
u_{n-r}\rangle=R$. Then $A_0=A_1\begin{bmatrix}I_r & 0 \\ 0 & A\end{bmatrix}$
for some $A\in\GL(n-r, q)$. In particular for any $A_0\in\GL(n, q)$ there exists a 
unique $A_1$ 
from $T$ such that $A_1^{-1}A_0$
is of the form $\begin{bmatrix}I_r & 0 \\ 0 & 
A\end{bmatrix}$.

\end{remark}

We are now ready to present the algorithm, followed by some implementation details.

\begin{description}
\item[Input.] Two $m$-alternating tuples $\bG=(G_1, \dots, G_m)$ and $\bH=(H_1, \dots, 
H_m)$ in $\Lambda(n, q)^m$ representing $m$-alternating spaces $\cG, \cH\leq 
\Lambda(n, q)$, respectively.
$m=cn$ for some constant $c$, and $n$ is large enough (larger than some fixed 
function of $c$).
\item[Output.] Either certify that $\cG$ does not satisfy $F(n, m, q, r)$, or a 
set 
$S$ consisting of all isometries between $\cG$ and $\cH$. (If $S=\emptyset$ then 
$\cG$ and $\cH$ are not isometric.)
\item[Algorithm procedure.] \ 
\begin{enumerate}
\item $S=\emptyset$.
\item Set $r\in \N$ such 
that $r\geq 4\cdot \frac{n-r}{m}$ if $n-r\geq m$, and $r\geq 4\cdot \frac{m}{n-r}$ 
if $m\geq n-r$.
\item Construct $\bB\in M((n-r)\times m, q)^r$ as described in 
Section~\ref{sec:outline} or before Definition~\ref{def:property_F_prime}.
\item Compute a linear basis of $\Adj(\bB)\leq M(n-r, q)\oplus M(m, q)$.
\item Let 
$\pi_1$ be the projection of $M(n-r, q)\oplus M(m, q)$ to $M(n-r, q)$ along $M(m, 
q)$. If 
$\dim(\pi_1(\Adj(\bB)))>n-r$, then return ``$\cG$ does not satisfy $F(n, m, q, 
r)$.''
\item List all $r$-individualisations in $\F_q^n$ by the algorithm in 
Proposition~\ref{prop:ind}. 
For every $r$-individualisation $\F_q^n=L\oplus R$ with an ordered basis $(v_1, \dots, v_r)$ of $L$, let $A_1$ be its corresponding invertible matrix produced by the algorithm. Do the following.
\begin{enumerate}
\item Construct $\bC\in M((n-r)\times m, q)^r$ w.r.t. $L\oplus R$ and $(v_1, 
\dots, v_r)$.
\item Compute a linear basis of $\Adj(\bB, \bC):=\{(A, D)\in M(n-r, q)\oplus M(m, 
q) : A\bB=\bC D\}$.
\item If $\dim(\pi_1(\Adj(\bB,\bC)))>n-r$, go to the next $r$-individualisation.
\item If $\dim(\pi_1(\Adj(\bB, \bC)))\leq n-r$, do the following: 
\begin{enumerate}
\item For every $A\in\pi_1(\Adj(\bB, \bC))$, if $A$ is invertible, let 
$A_2=\begin{bmatrix} I_r & 0\\0 & A\end{bmatrix}$. Test whether $A_0=A_2A_1^{-1}$ 
is an isometry between $\cG$ and $\cH$. If so, add $A_0$ to $S$.
\end{enumerate}
\end{enumerate}
\item Return $S$.
\end{enumerate}
\end{description}

We describe some implementation details.
\begin{description}
\item[Step 3.] $\bB$ is constructed by taking the upper-right $r\times (n-r)$ 
corners of $G_i$'s to get an $m$-matrix tuple $\bB'\in M(r\times (n-r), q)^m$, and 
flipping $\bB'$ to 
obtain an $r$-matrix tuple $\bB\in M((n-r)\times m, q)^r$. See also 
Figure~\ref{fig:flip} and~\ref{fig:slice}.
\item[Step 6.a.] $\bC$ is constructed as follows. In Step 6, by fixing an 
$r$-individualisation, we obtain a change-of-basis matrix $A_1=[v_1, \dots, v_r, 
u_1, \dots, u_{n-r}]$ as described in Proposition~\ref{prop:ind}. Let 
$\bH_1=\trans{A}_1\bH 
A_1$. Then perform the same 
procedure 
as in 
Step 3 for $\bH_1$.
\item[Step 6.d.i.] To test whether $A_0=A_2A_1^{-1}$ is an isometry 
between $\cG$ and $\cH$, we just need to test whether $\trans{A}_0\bG A_0$ and 
$\bH$ span the same alternating space.
\end{description}

It is straightforward to verify that the algorithm runs in time $q^{O(n)}$: the 
multiplicative cost of enumerating $r$-individualisation is at most $q^{2rn-r^2}$, and the 
multiplicative cost of enumerating $\pi_1(\Adj(\bB, \bC))$ is at most $q^{n-r}$. All other 
steps are basic tasks in linear algebra so can be carried out efficiently.

When $m=cn$ and $n$ larger than a fixed function of $c$, all but at most 
$1/q^{\Omega(n)}$ fraction of $\cG\leq \Lambda(n, q)$ satisfy 
$F(n, m, q, r)$ by Propositions~\ref{prop:key}, \ref{prop:stable}, 
and~\ref{prop:models}. Note that $\Omega(n)$ hides a constant depending on $c$.

To see the correctness, first note that by the test step in Step 6.d.i, only 
isometries will be added to $S$. So we need to argue that if $\cG$ is in $F(n, m, 
q, r)$, then every isometry 
$A_0\in\Iso(\cG,\cH)$ will be added to $S$. Recall that $A_0\in 
\GL(n, q)$ is an isometry from $\cG$ to $\cH$ if and only if there exists $D\in 
\GL(m, q)$ such that $\trans{A_0}\bG A_0=\bH^D$, which is equivalent to 
$\bG=A_0^{-\mathrm{t}}(\bH^D)A_0^{-1}$.
By Remark~\ref{rem:ind}, $A_0^{-1}\in\GL(n, q)$ can be written 
uniquely 
as $A_0^{-1}=A_1A_2^{-1}$ where $A_1$ is from the list $T$ produced by 
Proposition~\ref{prop:ind}, and $A_2^{-1}=\begin{bmatrix}I_r & 0 \\ 0 & A^{-1}
\end{bmatrix}$ for some invertible $A\in\GL(n-r,q)$.
When 
enumerating the individualisation corresponding to $A_1$, we have $\trans{A}_2\bG 
A_2=A_1^{\mathrm{t}}(\bH^D)A_1=(A_1^{\mathrm{t}}\bH A_1)^D$, which 
implies that $(A, D)\in \Adj(\bB, \bC)$ and $A\in \pi_1(\Adj(\bB, \bC))$. 
Since $A\bB = \bC D$ for some invertible $A$ and 
$D$, we have $\dim(\Adj(\bB, 
\bC))=\dim(\Adj(\bB))$ and $\dim(\pi_1(\Adj(\bB,\bC)))=\dim(\pi_1(\Adj(\bB)))$, 
which justifies Step 6.c together with the condition 
already imposed 
in Step 5. Since $A\in \pi_1(\Adj(\bB, \bC))$, it will be encountered when 
enumerating $\Adj(\bB, \bC)$ in Step 6.d.i, so $A_0=A_2A_1^{-1}$ will be built 
and, 
after the 
verification 
step, added to $S$.

\section{Dynamic programming}\label{sec:dp}

In this section, given a matrix group $G\leq \GL(n, q)$, we view $G$ as a 
permutation group on the domain $\F_q^n$, so basic tasks like membership testing 
and pointwise transporter can be solved in time $q^{O(n)}$ by permutation group 
algorithms. Furthermore a generating set 
of $G$ of size $q^{O(n)}$ can also be obtained in time $q^{O(n)}$. These 
algorithms are classical and can be found in \cite{Luks90,seress2003permutation}.

As mentioned in Section~\ref{sec:intro}, for \GrI, Luks' dynamic programming 
technique \cite{Luks99} can improve the brute-force  
$n!\cdot \poly(n)$ time bound to the $2^{O(n)}$ time bound, which can 
be understood as replacing the number of permutations $n!$ 
with the number of subsets $2^n$. 

In our view, Luks' dynamic programming technique 
is most transparent when working with the subset transporter problem. Given a 
permutation group $P\leq S_n$ and $S, T\subseteq [n]$ of size $k$, this technique 
gives a $2^k\cdot \poly(n)$-time algorithm to compute $P_{S\to 
T}:=\{\sigma\in P : \sigma(S)=T\}$ \cite{BQ12}. To illustrate the idea in the 
matrix group 
setting, we 
start with the subspace transporter problem.
\begin{problem}[Subspace transporter problem]
Let $G\leq \GL(n, q)$ be given by a set of generators, and let $V$, $W$ be two subspaces 
of $\F_q^n$ of dimension $k$. The subspace transporter problem asks to compute the 
coset $G_{V\to W}=\{g\in G : g(V)=W\}$.
\end{problem}
The subspace transporter problem admits the following brute-force algorithm. Fix a 
basis $(v_1, \dots, v_k)$ of $V$, and enumerate all ordered basis
of $W$ at the 
multiplicative cost of $q^{k^2}$. For each ordered basis $(w_1, \dots, w_k)$ of 
$W$, 
compute the coset $\{g\in G : \forall i\in[k], g(v_i)=w_i\}$ by using a sequence 
of 
pointwise stabiliser algorithms. This gives an 
algorithm running in time $q^{k^2+O(n)}$. 
Analogous to the permutation group setting, we aim to replace $O(q^{k^2})$, the 
number of ordered basis of 
$\F_q^k$, 
with $q^{\frac{1}{4}k^2+O(k)}$, the number of subspaces in $\F_q^k$, via a dynamic 
programming technique. For this we first observe the following. 

\begin{observation}\label{obs:enumerate_subspace}
There exists a deterministic algorithm that enumerates all subspaces of $\F_q^n$, 
and for each subspace computes an ordered basis, in time $q^{\frac{1}{4}n^2+O(n)}$.
\end{observation}
\begin{proof}
For $d\in \{0, 1, \dots, n\}$, let $S_d$ be the number of dimension-$d$ subspaces 
of $\F_q^n$. The total number of subspaces in $\F_q^n$ is 
$S_0+S_1+\dots+S_n=q^{\frac{1}{4}n^2+O(n)}$. To 
enumerate 
all subspaces we proceed by induction on the dimension in an increasing order. The 
case $d=0$ is trivial. For 
$d\geq 1$, suppose all subspaces of dimension $d-1$, each with an ordered basis, 
are listed. To list all subspaces of dimension $d$, for each dimension-$(d-1)$ 
subspace $U'$ with an ordered basis $(u_1, \dots, u_{d-1})$, for each vector 
$u_d\not\in U'$, form $U$ with the ordered basis $(u_1, \dots, u_d)$. Then test 
whether $U$ has been listed. If so discard it, and if not add $U$ together with 
this ordered basis to the list. The two for loops as above adds a multiplicative 
factor of at most $S_{d-1}\cdot q^{n}$, and other steps are basic linear 
algebra tasks. Therefore the total complexity is $\sum_{i=0}^nS_i\cdot 
q^{O(n)}=q^{\frac{1}{4}n^2+O(n)}$. 
\end{proof}

\begin{theorem}\label{thm:subsp_trans}
There exists a deterministic algorithm that solves the subspace transporter 
problem in time $q^{\frac{1}{4}k^2+O(n)}$. 
\end{theorem}
\begin{proof}
We fix an ordered basis $(v_1, \dots, v_k)$ of $V$, and for $d\in[k]$, let 
$V_d=\langle v_1, \dots, v_d\rangle$. 
The dynamic programming table is a list, indexed by 
subspaces $U\leq W$. For $U\leq W$ of dimension $d\in[k]$, the corresponding cell 
will store the coset
$G(V_d\to U)=\{g\in G:g(V_d)=U\}$. When $d=k$ the corresponding cell gives $G(V\to 
W)$. 

We fill in the dynamic programming table according to $d$ in an increasing order. 
For $d=0$ the problem is trivial. 
%
Now assume that for some $d\geq 1$, we have computed $G(V_l\to U')$ for all $0\leq 
l\leq 
d-1$ and subspace $U'\leq W$ of dimension $U$. To compute $G(V_d\to U)$ for some 
fixed $U\leq W$ of dimension $d$, note that any $g\in G(V_d\to U)$ has to map 
$V_{d-1}$ 
to some $(d-1)$-dimension subspace $U'\leq U$, and $v_d$ to 
some vector $u\in U\setminus U_0$. This shows that
$$G(V_d\to U)=\bigcup_{U'\leq U, \dim(U')=d-1}\bigcup_{u\in U\setminus 
U'}[G(V_{d-1}\to U')](v_d\to u).$$
To compute $[G(V_{d-1}\to U')](v_d\to u)$, we read $G(V_{d-1}\to U')$ from the 
table, then compute $[G(V_{d-1}\to U')](v_d\to u)$ using the pointwise transporter 
algorithm. The number of $u$ in $U\setminus U'$ is no more than $q^d$, and the 
number of $(d-1)$-dimension subspaces of $U$ is also no more than $q^d$. After 
taking these two unions, apply Sims' method to get a generating set of size 
$q^{O(n)}$. Therefore for each cell the time complexity is $q^{2d}\cdot 
q^{O(n)}=q^{O(n)}$. Therefore the whole dynamic programming table can be filled in 
time $q^{\frac{1}{4}k^2+O(k)}\cdot q^{O(n)}=q^{\frac{1}{4}k^2+O(n)}$.
%
%
\end{proof}

To apply the above idea to \AltMatSpIso, we will need to deal with 
the following problem. 


\begin{problem}[Alternating matrix transporter problem]
Let $H\leq \GL(n, q)$ be given by a set of generators, and let $A, B\in \Lambda(n, 
q)$ be two alternating matrices. The alternating matrix transporter 
problem asks to compute the coset $H_{A \to B}=\{g\in H: \trans{g}Ag=B\}$. 
\end{problem}

\begin{theorem}\label{thm:alt_form}
There exists a deterministic algorithm that solves the alternating matrix 
transporter 
problem in time $q^{\frac{1}{4}n^2+O(n)}$. 
\end{theorem}

\begin{proof}
Let $(e_1,\dots,e_n)$ be the standard basis vectors of $\F_q^n$, and let 
$E_d=\langle e_1,\dots,e_d\rangle$. For an alternating matrix $B$, and an ordered 
basis $(u_1, \dots, u_d)$ of a dimension-$d$ $U\leq \F_q^n$, $B|_U$ denotes the 
$d\times d$ alternating matrix $[u_1, \dots, u_d]^tB[u_1, \dots, u_d]$, called the 
restriction of $B$ to $U$. For a vector $v$ and $U$ with the ordered basis as 
above, $B_{U\times v}=[u_1, \dots, u_d]^tB v\in \F_q^d$.

Then we construct a dynamic programming table, which is a list indexed by all 
subspaces of $\F_q^n$. Recall that each subspace also comes with an ordered basis 
by 
Observation~\ref{obs:enumerate_subspace}. For any $U\leq 
\F_q^n$ of dimension $k$, its corresponding cell will store the coset
\begin{equation}
H(A|_{E_k}\to B|_{U})=\{g\in H: g(E_k)=U,~\trans{g}(A|_{E_k})g=B|_{U}\}.
\end{equation}

We will also fill in this list in the increasing order of the dimension $d$. The 
base case $d=0$ is trivial. 
Now, assume we have already compute $H(A|_{E_l}\to B|_{U'})$ for all $0\leq l\leq 
d-1$ and subspace $U'\leq \F_q^n$ of dimension $l$. To compute $H(A|_{E_k}\to 
B|_{U})$ for some $U\leq \F_q^n$ of dimension $d$, note that any $g\in 
H(A|_{E_d}\to B|_{U})$ satisfies the following. Firstly, $g$ sends $E_{d-1}$ to 
some dimension-$(d-1)$ subspace $U'\leq 
U$, and $A|_{E_{d-1}}\in \Lambda(d-1, q)$ to $B|_{U'}\in \Lambda(d-1, q)$. 
Secondly, $g$ sends $e_d$ to some vector 
$u\in U\setminus U'$, and $A|_{E_{d-1}\times e_d}\in \F_q^{d-1}$ to $B|_{U'\times 
u}\in \F_q^{d-1}$. 
This shows that 
\begin{equation}
H(A|_{E_d}\to B|_{U})=\bigcup_{U'\leq U,\dim(U')=d-1}\bigcup_{u\in U\setminus 
U'}\big[[H(A|_{E_{d-1}}\to B|_{U'})](e_d\to u)\big](A|_{E_{d-1}\times e_d}\to 
B|_{U'\times 
u}).
\end{equation}
To compute $\big[[H(A|_{E_{d-1}}\to B|_{U'})](e_d\to u)\big](A|_{E_{d-1}\times 
e_d}\to 
B|_{U'\times 
u})$, we read $H(A|_{E_{d-1}}\to B|_{U'})$ from the table, compute 
$[H(A|_{E_{d-1}}\to B|_{U'})](e_d\to u)$ using the pointwise transporter 
algorithm. As $[H(A|_{E_{d-1}}\to B|_{U'})](e_d\to u)$ induces an action 
on $\F_q^{d-1}$ corresponding to the last column of $A|_{E_d}$ with the last entry 
(which is $0$) removed, 
$\big[[H(A|_{E_{d-1}}\to B|_{U'})](e_d\to u)\big](A|_{E_{d-1}\times e_d}\to 
B|_{U'\times u})$ can be computed by another pointwise transporter algorithm. 
As in Theorem~\ref{thm:subsp_trans}, we go over the two unions and apply Sims' 
method to obtain a generating set of size $q^{O(n)}$. The time complexity for 
filling in each cell is seen to be $q^{2d}\cdot q^{O(n)}$, and the total time 
complexity is then $q^{\frac{1}{4}n^2+O(n)}$.
%
\end{proof}

We are now ready to prove Theorem~\ref{thm:minor}.

\paragraph{Theorem~\ref{thm:minor}, restated.}
Given $\bG=(G_1, \dots, G_m)$ and $\bH=(H_1, \dots, H_m)$ in $\Lambda(n, q)^m$ 
representing $m$-alternating spaces 
$\cG, \cH\leq \Lambda(n, 
q)$, there exists a deterministic algorithm for \AltMatSpIso in time 
$q^{\frac{1}{4}(m^2+n^2)+O(m+n)}$.
\begin{proof}
Let $(e_1,\dots,e_m)$ be the standard basis of $\F_q^m$, and let $E_k=\langle 
 e_1,\dots, e_k\rangle$. 
$v=(a_1, \dots, a_m)^{\mathrm{t}}\in \F_q^m$, define 
$\bH^v:=\sum_{i\in[m]}a_iH_i\in \Lambda(n, q)$. 
For a dimension-$k$ subspace $V\leq \F_q^m$ with an ordered basis $(v_1, \dots, 
v_k)$, $\bH^V:=(\bH^{v_1}, \dots, \bH^{v_k})\in \Lambda(n, q)^k$. 

The dynamic programming table is indexed by subspaces of 
$\F_q^m$, so the number of cells is no more than $q^{\frac{1}{4}m^2+O(m)}$. The 
cell corresponding to a dimension-$k$ subspace $V$ stores the coset
\begin{equation}
\Iso(\bG^{E_k}, \bH^{V})=\{(g, h)\in \GL(n,q)\times \GL(k, q): 
\trans{g}(\bG^{E_k})g=(\bH^V)^h\},
\end{equation}

We will fill in the dynamic programming table in the increasing order of the 
dimension $d$. Recall that each subspace also comes with an ordered basis 
by 
Observation~\ref{obs:enumerate_subspace}. The base case $d=0$ is trivial. 
Now assume we have computed $\Iso(\bG^{E_{\ell}}, \bH^V)$ for all $1\leq \ell\leq 
d-1$ and $V\leq \F_q^n$ of dimension $\ell$. To compute $\Iso(\bG^{E_d}, \bH^V)$ 
for $V\leq \F_q^n$ of dimension $d$, note that any $h$ in $(g, h)\in 
\Iso(\bG^{E_d}, \bH^V)$ satisfies the following. Firstly, $h$ sends $E_{d-1}$ to 
some dimension-$(d-1)$ subspace $V'\leq V$, and $(g, h)\in \Iso(\bG^{E_{d-1}}, 
\bH^{V'})$. Secondly, $h$ sends $e_k$ to some $v\in V\setminus V'$, and $g$ sends 
$\bG^{e_d}$ to $\bH^{v}$. This shows that 
$$
\Iso(\bG^{E_d}, \bH^V)=\bigcup_{V'\leq V,\dim(V')=d-1}\bigcup_{v\in V\setminus V'} 
\big[[\Iso(\bG^{E_{d-1}}, \bH^{V'})](e_d\to v)\big](\bG^{e_d} \to \bH^{v}).
$$
To compute $\big[[\Iso(\bG^{E_{d-1}}, \bH^{V'})](e_d\to v)\big](\bG^{e_d} \to 
\bH^{v})$, $\Iso(\bG^{E_{d-1}}, \bH^{V'})$ can be read from the table.
$[\Iso(\bG^{E_{d-1}}, \bH^{V'})](e_d\to v)$ is an instance of the pointwise 
transporter problem of $\GL(n, q)\times \GL(k, q)$ acting on $\F_q^m$, which can 
be solved in time $q^{O(m)}$. Finally 
$\big[[\Iso(\bG^{E_{d-1}}, \bH^{V'})](e_d\to v)\big](\bG^{e_d} \to 
\bH^{v})$ is an instance of the alternating matrix transporter problem, which can 
be solved, by Theorem~\ref{thm:alt_form}, in time $q^{\frac{1}{4}n^2+O(n)}$. Going 
over the two unions adds a multiplicative factor of $q^{2d}$, and then we 
apply Sims' method to reduce the generating set size to $q^{O(n)}$. Therefore for 
each cell the time complexity is $q^{2d}\cdot 
q^{\frac{1}{4}n^2+O(n+m)}=q^{\frac{1}{4}n^2+O(m+n)}$. Therefore the whole dynamic 
programming table can be filled in in time $q^{\frac{1}{4}m^2+O(m)}\cdot 
q^{\frac{1}{4}n^2+O(n+m)}=q^{\frac{1}{4}(n^2+m^2)+O(n+m)}$.
%
%
\end{proof}

\section{Discussions and future directions}\label{sec:discussion}

\subsection{Discussion on the prospect of worst-case time complexity of 
\AltMatSpIso}\label{subsec:discussion_algo}

While our main result is an average-case algorithm, we believe that the ideas 
therein suggest that an algorithm for \AltMatSpIso in time 
$q^{O(n^{2-\epsilon})}$
may be within reach. 

For this, we briefly recall some fragments of 
the history of \GrI, with a focus on the worst-case time complexity aspect. 
Two (families of) algorithmic ideas have been most responsible for the worst-case 
time 
complexity improvements for \GrI. The first idea, which we call the combinatorial  
idea, is to use certain combinatorial techniques including individualisation, 
vertex or edge 
refinement, and more generally the Weisfeiler-Leman refinement \cite{WL68}. The 
second idea, which we call the group theoretic idea, is to reduce \GrI to certain 
problems in permutation group 
algorithms, and then settle those problems using group theoretic techniques 
and structures. A major breakthrough utilising the group theoretic idea is the 
polynomial-time algorithm for graphs with bounded degree by 
Luks \cite{Luk82}. 
%

Some combinatorial techniques have been implemented and used in 
practice \cite{MP14}, though the worst-case analysis usually does not favour such 
algorithms 
(see e.g. \cite{CFI92}). On the other hand, while group theoretic algorithms for 
\GrI more than often come 
with a rigorous analysis, such algorithms usually only work with a restricted 
family of graphs (see e.g. \cite{Luk82}). The
major improvements on the worst-case time complexity of \GrI almost always rely on 
both ideas. The recent breakthrough, a quasipolynomial-time 
algorithm for \GrI by Babai \cite{Bab16,Bab17}, is a clear evidence. Even the 
previous record, a $2^{\tilde{O}(\sqrt{n})}$-time algorithm by Babai and Luks 
\cite{BL83}, relies on both Luks' group theoretic framework \cite{Luk82} and 
Zemlyachenko's combinatorial partitioning lemma \cite{ZKT85}. 

Let us return to \AltMatSpIso. It is clear that \AltMatSpIso can be studied in the 
context of matrix groups over finite fields. Computing with finite matrix groups 
though, 
turns out to be much more difficult than working with permutation groups. The 
basic 
constructive membership testing task subsumes the discrete log problem, and even 
with a number-theoretic oracle, a randomised polynomial-time algorithm for 
constructive membership testing was only 
recently obtained by Babai, Beals and 
Seress \cite{BBS09} for odd $q$. However, if a $q^{O(n+m)}$-time 
algorithm for $\AltMatSpIso$ is the main concern, then we can view 
$\GL(n, q)$ acting on the domain $\F_q^n$ of size $q^n$, so basic tasks like 
constructive membership testing are not a bottleneck. In addition, a 
group theoretic framework for matrix groups in vein of the corresponding 
permutation group results in \cite{Luk82} has also been developed by Luks 
\cite{Luk92}. Therefore, if we aim at a $q^{O(n+m)}$-time 
algorithm for \AltMatSpIso, the group theoretic aspect is 
relatively developed.

Despite all the results on the group theoretic aspect, as described in 
Section~\ref{subsec:bg}, a $q^{O(n+m)}$-time algorithm for 
\AltMatSpIso has been widely regarded to be very difficult, as such an 
algorithm would imply 
an algorithm that tests isomorphism of $p$-groups of class $2$ and 
exponent $p$ in time polynomial in the group order. Reflecting back on how 
the time complexity of \GrI has been improved, we realised that the other major 
idea, namely the combinatorial refinement idea, seemed missing in the context of 
\AltMatSpIso. By adapting the individualisation technique, developing an 
alternative route to the refinement step as used 
in \cite{BES80}, and demonstrating its usefulness in 
the linear algebraic 
Erd\H{o}s-R\'enyi model, we believe that this opens the door to systematically 
examine and adapt such 
combinatorial refinement techniques for \GrI to improve the worst-case time 
complexity of \AltMatSpIso. We mention one possibility 
here. In \cite{Qia17}, a notion of degree for alternating matrix spaces will be 
introduced, and 
it will be interesting to combine that degree notion with Luks' group theoretic 
framework for matrix groups \cite{Luk92} to see whether one can obtain a 
$q^{O(n+m)}$-time algorithm to test isometry of alternating matrix 
spaces with bounded degrees. If this is feasible, then one can try to develop a 
version of Zemlyachenko's combinatorial partition lemma for \AltMatSpIso in the 
hope to obtain a moderately exponential-time algorithm (e.g. in time 
$q^{O(n^{2-\epsilon})}$) for \AltMatSpIso.
%


\subsection{Discussion on the linear algebraic Erd\H{o}s-R\'enyi 
model}\label{subsec:discuss_er}

As far as we are aware, the linear algebraic Erd\H{o}s-R\'enyi model 
(Model~\ref{model:linER}) has not been discussed in the literature. We believe 
that 
this model may lead to some interesting mathematics. In this section we put some 
general remarks on this model. We will consider $\LinER(n, m, q)$, or the 
corresponding bipartite version of $\LinER$, $\BipLinER(n\times n, m, q)$ as 
defined in Section~\ref{subsec:random_bip}.

To start with, it seems to us reasonable to consider 
an event $E$ as happening with high probability only when $\Pr[E]\geq 
1-1/q^{\Omega(n)}$. To illustrate the reason, consider 
$\BipLinER(n\times n, m, q)$ with the following property $E(n, m, q)$. For a 
dimension-$m$ 
$\cB\leq M(n\times n, q)$, $\cB$ satisfies $E(n, m, q)$ if and only if for every 
$U\leq 
\F_q^n$, $\dim(\cB(U))\geq \dim(U)$. This corresponds to the concept of 
semi-stable as in the geometric invariant theory; compare with the stable concept 
as described in Section~\ref{subsec:outline}. One can think of 
$\cB$ being semi-stable as having a perfect matching \cite{GGOW16,IQS16,IQS17}. 
When $m=1$, $\cB=\langle B\rangle$ 
is semi-stable if and only if $B$ is invertible, so $1-\frac{1}{q}\geq \Pr[E(n, 1, 
q)]\geq 
1-\frac{1}{q-1}$. On the other hand when $m=4$, since stable implies semi-stable, 
from Section~\ref{subsec:outline} we have $\Pr[E(n, 4, q)]\geq 
1-\frac{1}{q^{\Omega(n)}}$. So though $E(n, 1, q)$ happens with some nontrivial 
probability, it seems not fair to consider $E(n, 1, q)$ happens with high 
probability, while $E(n, 4, q)$ should be thought of as happening with high 
probability. 

The above example suggests that the phenomenon in the linear 
algebraic Erd\H{o}s-R\'enyi model can be different from its classical 
correspondence. Recall that in the classical Erd\H{o}s-R\'enyi 
model, an important discovery is that most properties $E$ have a threshold $m_E$. 
That is, when the edge number $m$ is slightly less than $m_E$, then $E$ almost 
surely does not happen. On the other hand, if $m$ surpasses $m_E$ slightly then 
$E$ almost surely happens. $m_E$ is usually a nonconstant function of the vertex 
number, as few interesting things can happen when we have only a constant number 
of edges. However, the above example about the semi-stable property suggests that, 
if there is a threshold for this property, then this threshold has to be between 
$1$ and $4$, as we have seen the transition from $1-1/q^{O(1)}$ to 
$1-1/q^{\Omega(n)}$ when $m$ goes from $1$ to $4$. This is not surprising though, 
as one ``edge'' in the linear setting is one matrix, which seems much more 
powerful 
than an edge in a graph. It should be possible to 
pin down the exact threshold for the semi-stable property, and we conjecture 
that the transition (from $1-1/q^{O(1)}$ to $1-1/q^{\Omega(n)}$) happens from $2$ 
to $3$ as this is where the transition from tame to wild as in the representation 
theory \cite[Chapter 4.4]{Ben98} happens for the representations of the Kronecker 
quivers. This hints on one research direction on \LinER, that is, to determine 
whether the threshold phenomenon happens with monotone properties. 

The research on \LinER has to depend on whether there are enough interesting 
properties of matrix spaces. We mention two properties that originate from existing
literature; more properties can be found in the forthcoming paper \cite{Qia17}. 
Let 
$\cG$ be an 
$m$-alternating space in $\Lambda(n, q)$. For $U\leq \F_q^n$ of dimension $d$ with 
an ordered basis $(v_1, \dots, v_d)$, the restriction of $\cG$ to $U$ is defined 
as $\{ [v_1, \dots, v_d]^t G [v_1, \dots, v_d] : G \in \cG\}$ which is an 
alternating space in $\Lambda(d, q)$. The first property is 
the following. Let 
$s(\cG)$ be 
the smallest number for the existence of a dimension-$s$ subspace $U$ such that 
the restriction of $\cG$ to $U$ is of dimension $m$. 
This notion is one key to the upper bound on the number of 
$p$-groups \cite{Sim65,BNV07}. It is interesting to study the asymptotic behavior 
of $s(\cG)$.
The second property is the following. Call $U\leq \F_q^n$ an independent 
subspace, if the restriction of $\cG$ to $U$ is the zero space. We can define the 
independent number of $\cG$ accordingly. This mimics the 
independent sets for graphs, and seems to relate to the independent number concept 
for non-commutative graphs which are used to model quantum channels \cite{DSW13}. 
Again, it is interesting to study the asymptotic behavior of the independent 
number.

Finally, as suggested in \cite{DSW13} (where they consider Hermitian matrix spaces 
over $\C$), 
the model may be studied over infinite 
fields, where 
we replace ``with high probability'' with ``generic'' as in the algebraic 
geometry sense.

\subsection{Discussion on enumerating of 
$p$-groups of class $2$ and exponent $p$}\label{subsec:lb}

In this section we observe that Corollary~\ref{cor:main} can be used to slightly 
improve the upper bound on the number of $p$-groups of class $2$ and exponent $p$, 
as in \cite[Theorem 19.3]{BNV07}. (The proof idea there was essentially based on 
Higman's bound on the number of $p$-groups of Frattini class $2$ \cite{Hig60}.) We 
will outline the basic idea, and then focus on discussing how random graph 
theoretic ideas may be used to further improve on enumerating $p$-groups of 
Frattini class $2$. 

\subsubsection{From Corollary~\ref{cor:main} to enumerating $p$-groups of class 
$2$ and exponent $p$}

\begin{theorem}[\cite{Hig60,BNV07}]
The number of $p$-groups of class $2$ and exponent $p$ of order $p^\ell$ is upper 
bounded by $\ell\cdot p^{\frac{2}{27}\ell^3-\frac{2}{9}\ell^2+\frac{49}{72}\ell}$.
\end{theorem}

We will use Corollary~\ref{cor:main} to show that, $\frac{49}{72}$, the 
coefficient of the linear term on the exponent, can be decreased. For this we 
recall the proof idea as in \cite{Hig60,BNV07}.

Recall that if a $p$-group $G$ is 
of class $2$ and exponent $p$, then its commutator subgroup and commutator 
quotient can be identified as vector spaces over $\F_p$.  We say that a class-$2$ 
and 
exponent-$p$ $p$-group $G$ is of parameter $(n, m)$, if $\dim(G/[G,G])=n$, and 
$\dim([G,G])=m$. By using relatively free $p$-groups 
of class $2$ and exponent $p$ as how Higman used relatively free $p$-groups of 
Frattini class $2$ \cite{Hig60}, we have the following 
result in vein of \cite[Theorem 2.2]{Hig60}. 
\begin{theorem}\label{thm:criterion}
The number of $p$-groups of $\Phi$-class $2$ and parameter $(n, m)$ is equal to 
the number of orbits of the natural $\GL(n, p)$ action on the codimension-$m$ 
subspaces of 
$\Lambda(n, p)$.
\end{theorem}
Note that \cite[Theorem 2.2]{Hig60} needs $p$ to be odd, due to the complication 
caused by the Frattini class $2$ condition. 

We then need to translate the 
codimension-$m$ condition in 
Theorem~\ref{thm:criterion} to a dimension-$m$ condition.
\begin{observation}\label{obs:dual}
The number of orbits of the $\GL(n, p)$ action on the codimension-$m$ subspaces of 
$\Lambda(n, p)$ is equal to the number of orbits of this action on the 
dimension-$m$ 
subspaces of $\Lambda(n, p)$.
\end{observation}
\begin{proof}
Define the standard bilinear form $P$ on $\Lambda(n, p)$ by $P(A, 
B)=Tr(\trans{A}B)$, which gives a bijective map between dimension-$m$ 
subspaces and 
codimension-$m$ subspaces of $\Lambda(n, p)$. It remains to verify that it yields 
a 
bijection between the $\GL(n, p)$ orbits as well. For this we check the following. 
Suppose $P( A,  
B)=0$. For any $X\in\GL(n, p)$, $P(\trans{X}AX, 
X^{-1}BX^{-\mathrm{t}})=
Tr(\trans{X}\trans{A}XX^{-1}BX^{-\mathrm{t}})
=Tr(\trans{A}B)=P(
 A, B)=0$. This implies that a dimension-$m$ subspace $\cB$ and a 
 codimension-$m$ subspace $\cC$ are orthogonal to each other w.r.t. 
 $P$, if and only if $\cB^X$ and $\cC^{X^{-\mathrm{t}}}$ are orthogonal 
 to each other 
 w.r.t. $P$. This concludes the proof.
\end{proof}

Therefore we reduce to study the number of orbits of $\GL(n, p)$ on dimension-$m$ 
subspaces in $\Lambda(n, 
p)$. Recall that $\dim(\Lambda(n, p))=n(n-1)/2$. 

Suppose we want to upper bound the number of $p$-groups of order $p^\ell$ of 
parameter $(n, m)$ when $m=cn$ for some constant $c$ ($\ell=m+n$). The 
number of dimension-$m$ subspaces of $\Lambda(n, p)$ is 
$\gbinom{\binom{n}{2}}{m}{p}$. Therefore a trivial upper bound is just to assume 
that every orbit is as small as possible, which gives 
$\gbinom{\binom{n}{2}}{m}{p}$ as the upper bound. Now that we have 
Corollary~\ref{cor:main}, by the orbit-stabilizer theorem, we know 
$(1-\frac{1}{p^{\Omega(n)}})$ fraction of the 
subspaces lie in an orbit of size $p^{n^2-O(n)}$. On the other hand, for a 
$\frac{1}{p^{\Omega(n)}}$ fraction of the subspaces, we have no control, so we 
simply assume that there each orbit is of size $1$. Summing over the two parts, we 
obtain an upper bound  
$$\frac{\gbinom{\binom{n}{2}}{m}{p}}{p^{n^2-O(n)}}+
\frac{\gbinom{\binom{n}{2}}{m}{p}}{p^{\Omega(n)}}$$
on the number of such $p$-groups. Note that the first summand will be dominated by 
the 
second one when $n$ is large enough. Plugging this into Higman's 
argument, we can show that the coefficient of the 
linear term is smaller than $\frac{49}{72}$. This idea can be generalised to 
deal with $p$-groups of Frattini class $2$ without much difficulty.

\subsubsection{Discussion on further improvements}

Our improvement on the upper bound of the number of $p$-groups of class $2$ and 
exponent $p$ is very modest. But this opens up the possibility of transferring 
random graph theoretic ideas to study enumerating such $p$-groups. In particular, 
this suggests the similarity between the number of unlabelled graphs with $n$ 
vertices and $m$ edges, and the 
number of $p$-groups of class $2$ and exponent $p$ of parameter $(n, m)$. 

A celebrated result from 
random grapth theory suggests that the number of unlabelled graphs with $n$ 
vertices and $m$ edges is $\sim \binom{\binom{n}{2}}{m}/n!$ \cite{Wri71} 
(see also \cite[Chapter 9.1]{Bol01}) when $cn\log n\leq m\leq \binom{n}{2}-cn\log 
n$. Note that this result implies, and is considerably stronger than, that most 
graphs have the trivial automorphism group. It is then tempting to explore whether 
the idea in \cite{Wri71} can be adapted to show that when $m=cn$, the number 
of $p$-groups of class $2$ and exponent $p$ of parameter $(n, m)$ is $\sim 
\binom{\binom{n}{2}}{m}/|\mathrm{PSL}(n, p)|$. If this was true, then it would 
imply that the number of $p$-groups of class $2$ and exponent $p$ of order 
$p^\ell$ is upper bounded by $p^{\frac{2}{27}\ell^3-\frac{2}{3}\ell^2+O(\ell)}$. 
This would then match the coefficient of the quadratic term on the exponent of the 
lower bound, answering \cite[Question 22.8]{BNV07} in the case of such groups. 
Further implications to $p$-groups of Frattini class $2$ should follow as well.

\appendix

\section{\AltMatSpIso and $p$-group isomorphism testing}\label{app:eq} 

The content in this appendix is classical by \cite{Bae38,War76}. See also 
\cite{Wil09} and \cite{GQ14}.

Suppose we are given two $p$-groups of class $2$ and exponent $p$, $G_1$ and $G_2$ 
of order $p^\ell$. For $G_i$, let $b_i:G_i/[G_i,G_i]\times G_i/[G_i,G_i]\to 
[G_i,G_i]$ be the 
commutator map where $[G_i,G_i]$ denotes the commutator subgroup. By the class $2$ 
and exponent $p$ assumption, $G_i/[G_i,G_i]$ are elementary abelian groups of 
exponent $p$. For $G_1$ and $G_2$ to be isomorphic it is necessary that 
$[G_1,G_1]\cong [G_2,G_2]\cong \Z_p^m$ and $G_1/[G_1,G_1]\cong G_2/[G_2,G_2]\cong 
\Z_p^n$ such that $m+n=\ell$. Furthermore $b_i$'s are alternating 
bilinear maps. So we 
have alternating bilinear maps $b_i:\F_p^n\times \F_p^n\to\F_p^m$. $G_1$ and $G_2$ 
are isomorphic if and only if there exist $A\in\GL(n, p)$ and $D\in\GL(n, p)$ 
such that for every $u, v\in \F_p^n$, $b_1(A(u), A(v))=D(b_2(u, v))$. Representing 
$b_i$ as a tuple of alternating matrices $\bB_i=(B_1, \dots, B_m)\in \Lambda(n, 
p)^m$, it translates to ask whether $\trans{A}\bB_1A=\bB_2^D$. Letting $\cB_i$ be 
the linear span of $\bB_i$, this becomes an instance of \AltMatSpIso w.r.t. 
$\cB_1$ and $\cB_2$.

When $p>2$, we can reduce \AltMatSpIso to isomorphism testing of $p$-groups of 
class $2$ and exponent $p$ using the following construction. Starting from $\bG\in 
\Lambda(n, p)^m$ representing $\cG$, $\bG$ can be viewed as representing a bilinar 
map $b:\F_p^n\times 
\F_p^n\to \F_p^m$. Define a group $G$ with operation $\circ$ over the set 
$\F_p^m\times \F_p^n$ as
$
(v_1, u_1)\circ (v_2, u_2)=(v_1+v_2+\frac{1}{2}b(u_1,u_2), u_1+u_2).
$
It can be verified that $G$ is a $p$-group of class $2$ and exponent $p$, and it 
is known that two such groups $G_1$ and $G_2$ built from $\bG_1$ and $\bG_2$ are 
isomorphic if and only if $\cG_1$ and $\cG_2$ are isometric. 

When working with groups in the Cayley table model, and working with \AltMatSpIso 
in time $p^{O(m+n)}$, 
the above procedures can be performed efficiently. In \cite{BMW15} it is discussed 
which models of computing with finite groups admit the reduction from isomorphism 
testing of $p$-groups of class $2$ and exponent $p$ to the pseudo-isometry testing 
of alternating bilinear maps. In particular it is concluded there that the 
reduction works in the permutation group quotient model introduced in \cite{KL90}.

\paragraph{Acknowledgement.} We thank G\'abor Ivanyos and James B. 
Wilson for 
helpful discussions and useful information. Y. Q. was supported by ARC DECRA 
DE150100720 during this work.

\bibliographystyle{alpha}
\bibliography{references}

\begin{thebibliography}{GGOW16}

\bibitem[Bab79]{Bab79}
L{\'{a}}szl{\'{o}} Babai.
\newblock {Monte-Carlo} algorithms in graph isomorphism testing.
\newblock Technical Report 79-10, D\'{e}p. Math. et Stat., Universit{\'e} de
  Montr{\'e}al, 1979.

\bibitem[Bab16]{Bab16}
L{\'{a}}szl{\'{o}} Babai.
\newblock Graph isomorphism in quasipolynomial time [extended abstract].
\newblock In {\em Proceedings of the 48th Annual {ACM} {SIGACT} Symposium on
  Theory of Computing, {STOC} 2016, Cambridge, MA, USA, June 18-21, 2016},
  pages 684--697, 2016.
\newblock arXiv:1512.03547, version 2.

\bibitem[Bab17]{Bab17}
L{\'{a}}szl{\'{o}} Babai.
\newblock {Fixing the UPCC case of Split-or-Johnson}.
\newblock \url{http://people.cs.uchicago.edu/~laci/upcc-fix.pdf}, 2017.

\bibitem[Bae38]{Bae38}
Reinhold Baer.
\newblock Groups with abelian central quotient group.
\newblock {\em Transactions of the American Mathematical Society},
  44(3):357--386, 1938.

\bibitem[BBS09]{BBS09}
L{\'{a}}szl{\'{o}} Babai, Robert Beals, and {\'{A}}kos Seress.
\newblock Polynomial-time theory of matrix groups.
\newblock In {\em Proceedings of the 41st Annual {ACM} Symposium on Theory of
  Computing, {STOC} 2009, Bethesda, MD, USA, May 31 - June 2, 2009}, pages
  55--64, 2009.

\bibitem[Ben98]{Ben98}
David~J. Benson.
\newblock Representations and cohomology. i, volume 30 of cambridge studies in
  advanced mathematics, 1998.

\bibitem[BES80]{BES80}
L{\'{a}}szl{\'{o}} Babai, Paul Erd{\H{o}}s, and Stanley~M. Selkow.
\newblock Random graph isomorphism.
\newblock {\em {SIAM} J. Comput.}, 9(3):628--635, 1980.

\bibitem[BK79]{BK79}
L{\'{a}}szl{\'{o}} Babai and Ludek Ku\v{c}era.
\newblock Canonical labelling of graphs in linear average time.
\newblock In {\em 20th Annual Symposium on Foundations of Computer Science, San
  Juan, Puerto Rico, 29-31 October 1979}, pages 39--46, 1979.

\bibitem[BL83]{BL83}
L{\'{a}}szl{\'{o}} Babai and Eugene~M. Luks.
\newblock Canonical labeling of graphs.
\newblock In {\em Proceedings of the 15th Annual {ACM} Symposium on Theory of
  Computing, 25-27 April, 1983, Boston, Massachusetts, {USA}}, pages 171--183,
  1983.

\bibitem[BL08]{BL08}
Peter~A. Brooksbank and Eugene~M. Luks.
\newblock Testing isomorphism of modules.
\newblock {\em Journal of Algebra}, 320(11):4020 -- 4029, 2008.

\bibitem[BMW15]{BMW15}
Peter~A. Brooksbank, Joshua Maglione, and James~B. Wilson.
\newblock A fast isomorphism test for groups of genus 2.
\newblock arXiv:1508.03033, 2015.

\bibitem[BNV07]{BNV07}
Simon~R. Blackburn, Peter~M. Neumann, and Geetha Venkataraman.
\newblock {\em Enumeration of finite groups}.
\newblock Cambridge Univ. Press, 2007.

\bibitem[BO08]{BO08}
Peter~A. Brooksbank and E.~A. O'Brien.
\newblock Constructing the group preserving a system of forms.
\newblock {\em Internat. J. Algebra Comput.}, 18(2):227--241, 2008.

\bibitem[Bol01]{Bol01}
B.~Bollob{\'a}s.
\newblock {\em Random Graphs}.
\newblock Cambridge Studies in Advanced Mathematics. Cambridge University
  Press, 2001.

\bibitem[BQ12]{BQ12}
L\'aszl\'o Babai and Youming Qiao.
\newblock Polynomial-time isomorphism test for groups with {A}belian {S}ylow
  towers.
\newblock In {\em 29th STACS}, pages 453 -- 464. Springer LNCS 6651, 2012.

\bibitem[BW12]{BW12}
Peter~A. Brooksbank and James~B. Wilson.
\newblock Computing isometry groups of {Hermitian} maps.
\newblock {\em Trans. Amer. Math. Soc.}, 364:1975--1996, 2012.

\bibitem[CFI92]{CFI92}
Jin{-}yi Cai, Martin F{\"{u}}rer, and Neil Immerman.
\newblock An optimal lower bound on the number of variables for graph
  identifications.
\newblock {\em Combinatorica}, 12(4):389--410, 1992.

\bibitem[CIK97]{CIK97}
Alexander Chistov, G\'{a}bor Ivanyos, and Marek Karpinski.
\newblock Polynomial time algorithms for modules over finite dimensional
  algebras.
\newblock In {\em Proceedings of the 1997 international symposium on Symbolic
  and algebraic computation}, ISSAC '97, pages 68--74, New York, NY, USA, 1997.
  ACM.

\bibitem[DSW13]{DSW13}
Runyao Duan, Simone Severini, and Andreas~J. Winter.
\newblock Zero-error communication via quantum channels, noncommutative graphs,
  and a quantum lov{\'{a}}sz number.
\newblock {\em {IEEE} Trans. Information Theory}, 59(2):1164--1174, 2013.

\bibitem[ELGO02]{ELO02}
Bettina Eick, C.~R. Leedham-Green, and E.~A. O'Brien.
\newblock Constructing automorphism groups of p-groups.
\newblock {\em Communications in Algebra}, 30(5):2271--2295, 2002.

\bibitem[ER59]{ER59}
P~Erd{\H{o}}s and A~R{\'e}nyi.
\newblock On random graphs.
\newblock {\em Publicationes Mathematicae Debrecen}, 6:290--297, 1959.

\bibitem[ER63]{ER63}
Paul Erd{\H{o}}s and Alfr{\'e}d R{\'e}nyi.
\newblock Asymmetric graphs.
\newblock {\em Acta Mathematica Hungarica}, 14(3-4):295--315, 1963.

\bibitem[FN70]{FN70}
V.~Felsch and J.~Neub\"user.
\newblock On a programme for the determination of the automorphism group of a
  finite group.
\newblock In Pergamon J.~Leech, editor, {\em Computational Problems in Abstract
  Algebra (Proceedings of a Conference on Computational Problems in Algebra,
  Oxford, 1967)}, pages 59--60, Oxford, 1970.

\bibitem[GGOW16]{GGOW16}
Ankit Garg, Leonid Gurvits, Rafael Oliveira, and Avi Wigderson.
\newblock A deterministic polynomial time algorithm for non-commutative
  rational identity testing.
\newblock In {\em {IEEE} 57th Annual Symposium on Foundations of Computer
  Science, {FOCS} 2016, 9-11 October 2016, Hyatt Regency, New Brunswick, New
  Jersey, {USA}}, pages 109--117, 2016.

\bibitem[GQ17a]{GQ14}
Joshua~A. Grochow and Youming Qiao.
\newblock Algorithms for group isomorphism via group extensions and cohomology.
\newblock {\em SIAM Journal on Computing}, 46(4):1153--1216, 2017.

\bibitem[GQ17b]{GQ17}
Joshua~A. Grochow and Youming Qiao.
\newblock Isomorphism problems in linear algebra.
\newblock In preparation, 2017.

\bibitem[GR16]{GR16}
Fran{\c{c}}ois~Le Gall and David~J. Rosenbaum.
\newblock On the group and color isomorphism problems.
\newblock {\em CoRR}, abs/1609.08253, 2016.

\bibitem[Hig60]{Hig60}
Graham Higman.
\newblock Enumerating $p$-groups. {I}: Inequalities.
\newblock {\em Proceedings of the London Mathematical Society}, 3(1):24--30,
  1960.

\bibitem[IKS10]{IKS10}
G{\'{a}}bor Ivanyos, Marek Karpinski, and Nitin Saxena.
\newblock Deterministic polynomial time algorithms for matrix completion
  problems.
\newblock {\em {SIAM} J. Comput.}, 39(8):3736--3751, 2010.

\bibitem[IQ17]{IQ17}
G{\'a}bor Ivanyos and Youming Qiao.
\newblock Algorithms based on $*$-algebras, and their applications to
  isomorphism of polynomials with one secret, group isomorphism, and polynomial
  identity testing.
\newblock {\em CoRR}, abs/1708.03495, 2017.

\bibitem[IQS16]{IQS16}
G{\'a}bor Ivanyos, Youming Qiao, and K.~V. Subrahmanyam.
\newblock Non-commutative {E}dmonds' problem and matrix semi-invariants.
\newblock {\em {Computational Complexity}}, pages 1--47, 2016.
\newblock In press; see \url{http://dx.doi.org/10.1007/s00037-016-0143-x}.

\bibitem[IQS17]{IQS17}
G{\'a}bor Ivanyos, Youming Qiao, and K.~V. Subrahmanyam.
\newblock Constructive non-commutative rank is in deterministic polynomial
  time.
\newblock In {\em the 8th Innovations in Theoretical Computer Science (ITCS)},
  2017.

\bibitem[Kar79]{Kar79}
Richard~M. Karp.
\newblock Probabilistic analysis of a canonical numbering algorithm for graphs.
\newblock In {\em Proceedings of the AMS Symposium in Pure Mathematics},
  volume~34, pages 365--378, 1979.

\bibitem[Kin94]{Kin94}
Alastair~D. King.
\newblock Moduli of representations of finite dimensional algebras.
\newblock {\em The Quarterly Journal of Mathematics}, 45(4):515--530, 1994.

\bibitem[KL90]{KL90}
William~M. Kantor and Eugene~M. Luks.
\newblock Computing in quotient groups.
\newblock In {\em Proceedings of the 22nd Annual {ACM} Symposium on Theory of
  Computing, May 13-17, 1990, Baltimore, Maryland, {USA}}, pages 524--534,
  1990.

\bibitem[KST93]{KST93}
Johannes K\"{o}bler, Uwe Sch\"{o}ning, and Jacobo Tor\'{a}n.
\newblock {\em The graph isomorphism problem: its structural complexity}.
\newblock Birkhauser Verlag, Basel, Switzerland, Switzerland, 1993.

\bibitem[KV12]{KV12}
Dmitry Kerner and Victor Vinnikov.
\newblock Determinantal representations of singular hypersurfaces in pn.
\newblock {\em Advances in Mathematics}, 231(3):1619--1654, 2012.

\bibitem[Lip78]{Lip78}
Richard~J. Lipton.
\newblock The beacon set approach to graph isomorphism.
\newblock Yale University. Department of Computer Science, 1978.

\bibitem[Luk82]{Luk82}
Eugene~M. Luks.
\newblock Isomorphism of graphs of bounded valence can be tested in polynomial
  time.
\newblock {\em J. Comput. Syst. Sci.}, 25(1):42--65, 1982.

\bibitem[Luk90]{Luks90}
Eugene~M. Luks.
\newblock Lectures on polynomial-time computation in groups.
\newblock {\em Lecture notes}, 1990.

\bibitem[Luk92]{Luk92}
Eugene~M. Luks.
\newblock Computing in solvable matrix groups.
\newblock In {\em 33rd Annual Symposium on Foundations of Computer Science,
  Pittsburgh, Pennsylvania, USA, 24-27 October 1992}, pages 111--120, 1992.

\bibitem[Luk99]{Luks99}
Eugene~M. Luks.
\newblock Hypergraph isomorphism and structural equivalence of boolean
  functions.
\newblock In {\em Proceedings of the Thirty-first Annual ACM Symposium on
  Theory of Computing}, STOC '99, pages 652--658, New York, NY, USA, 1999. ACM.

\bibitem[LW12]{LW12}
Mark~L. Lewis and James~B. Wilson.
\newblock Isomorphism in expanding families of indistinguishable groups.
\newblock {\em Groups - Complexity - Cryptology}, 4(1):73–110, 2012.

\bibitem[MFK94]{MFK94}
David Mumford, John Fogarty, and Frances Kirwan.
\newblock {\em Geometric invariant theory}.
\newblock Springer-Verlag, 1994.

\bibitem[Mil78]{Mil78}
Gary~L. Miller.
\newblock On the {$n \log n$} isomorphism technique (a preliminary report).
\newblock In {\em STOC}, pages 51--58, New York, NY, USA, 1978. ACM.

\bibitem[MP14]{MP14}
Brendan~D. McKay and Adolfo Piperno.
\newblock Practical graph isomorphism, {II}.
\newblock {\em Journal of Symbolic Computation}, 60(0):94 -- 112, 2014.

\bibitem[O'B93]{OBr93}
Eamonn~A. O'Brien.
\newblock Isomorphism testing for p-groups.
\newblock {\em Journal of symbolic computation}, 16(3):305--320, 1993.

\bibitem[Qia17]{Qia17}
Youming Qiao.
\newblock Matrix spaces as a linear algebraic analogue of graphs.
\newblock In preparation, 2017.

\bibitem[Ros13]{Ros13a}
David~J. Rosenbaum.
\newblock Bidirectional collision detection and faster deterministic
  isomorphism testing.
\newblock {\em arXiv preprint arXiv:1304.3935}, 2013.

\bibitem[Ser03]{seress2003permutation}
{\'A}kos Seress.
\newblock {\em Permutation group algorithms}, volume 152.
\newblock Cambridge University Press, 2003.

\bibitem[Sim65]{Sim65}
Charles~C. Sims.
\newblock Enumerating $p$-groups.
\newblock {\em Proceedings of the London Mathematical Society}, 3(1):151--166,
  1965.

\bibitem[War76]{War76}
Robert~B. Warfield.
\newblock {\em Nilpotent Groups}.
\newblock Number 513 in Lecture Notes in Mathematics; 513. Springer-Verlag,
  1976.

\bibitem[Wil09]{Wil09}
James~B. Wilson.
\newblock Decomposing $p$-groups via {J}ordan algebras.
\newblock {\em Journal of Algebra}, 322(8):2642--2679, 2009.

\bibitem[Wil14]{Wil14}
James~B. Wilson.
\newblock {2014 conference on \emph{Groups, Computation, and Geometry} at
  Colorado State University, co-organized by P. Brooksbank, A. Hulpke, T.
  Penttila, J. Wilson, and W. Kantor}.
\newblock Personal communication, 2014.

\bibitem[WL68]{WL68}
Boris Weisfeiler and Andrei~A. Leman.
\newblock A reduction of a graph to a canonical form and an algebra arising
  during this reduction.
\newblock {\em Nauchno-Technicheskaya Informatsia}, 2(9):12--16, 1968.

\bibitem[Wri71]{Wri71}
Edward~M Wright.
\newblock Graphs on unlabelled nodes with a given number of edges.
\newblock {\em Acta Mathematica}, 126(1):1--9, 1971.

\bibitem[ZKT85]{ZKT85}
Viktor~N. Zemlyachenko, Nickolay~M. Korneenko, and Regina~I. Tyshkevich.
\newblock Graph isomorphism problem.
\newblock {\em Journal of Soviet Mathematics}, 29(4):1426--1481, 1985.

\end{thebibliography}

\end{document}